\documentclass[11pt]{article}

\usepackage{amssymb, latexsym, mathtools, amsthm}
\begingroup
    \makeatletter
    \@for\theoremstyle:=definition,remark,plain\do{%
        \expandafter\g@addto@macro\csname th@\theoremstyle\endcsname{%
            \addtolength\thm@preskip\parskip
            }%
        }
\endgroup
\usepackage{enumerate}
\usepackage{pgf,tikz}
\usepackage{pdfsync}
\usepackage{txfonts}
\usepackage[T1]{fontenc}
\usepackage{booktabs}
 \usepackage{graphicx}
\definecolor{dnrbl}{rgb}{0,0,0.3}
\definecolor{dnrgr}{rgb}{0,0.3,0}
\definecolor{dnrre}{rgb}{0.5,0,0}
\usepackage[colorlinks=true, citecolor=dnrgr, linkcolor=dnrre, urlcolor=dnrre] {hyperref}
\usepackage{xcolor}
\usepackage{parskip}
\usepackage{tabularx, colortbl}
\usepackage{geometry}
 \usepackage[scriptsize, up]{caption}
\usepackage{pgf,tikz}
\usetikzlibrary{decorations, decorations.pathmorphing}
\usetikzlibrary {shapes}

\theoremstyle{plain}
\newtheorem{thm}{Theorem}[section]

\newtheorem{lem}[thm]{Lemma}

\newtheorem{defi}[thm]{Definition}
\newtheorem{obs}[thm]{Observation}
\newtheorem{ques}[thm]{Question}
\makeatletter

\let\c@table\c@figure
\makeatother

\renewenvironment{abstract}
 { \normalsize
  \list{}{
    \setlength{\leftmargin}{.0cm}%
    \setlength{\rightmargin}{\leftmargin}%
    }%
  \item {\bf \abstractname.} \relax}
 {\endlist}

\usepackage{amsthm}

\makeatletter
\newtheorem*{rep@theorem}{\rep@title}
\newcommand{\newreptheorem}[2]{%
\newenvironment{rep#1}[1]{%
 \def\rep@title{#2 \ref{##1}}%
 \begin{rep@theorem}}%
 {\end{rep@theorem}}}
\makeatother

\newreptheorem{defi}{Definition}

\title{The idemetric property: when most distances are   (almost) \newline the same}

\author{George Barmpalias$^{1}$ \and Neng Huang$^{2}$   \and Andrew Lewis-Pye$^{3}$ \and  Angsheng Li$^{4}$ \and  Xuechen Li$^{5}$ \and Yicheng Pan$^{4}$ \and Tim Roughgarden$^{6}$}
\date{January 15th, 2019}
\begin{document}
\maketitle
\begin{abstract}
We introduce the \emph{idemetric} property, which formalises the idea that most nodes in a graph have similar distances between them, and which turns out to be quite standard amongst small-world network models. Modulo reasonable sparsity assumptions, we are then able to show that a strong form of idemetricity is actually equivalent to a very weak expander condition (PUMP). This provides a direct way of providing short proofs that  small-world network models such as the Watts-Strogatz model  are strongly idemetric (for a wide range of parameters), and also provides further evidence that being idemetric is a common property.  

We then consider how satisfaction of the idemetric property is relevant to algorithm design.  
For idemetric graphs we observe, for example,  that a single breadth-first search provides a solution to the all-pairs shortest paths problem, so long as one is prepared to accept paths which are of stretch close to 2 with high probability. Since we are able to show that Kleinberg's model is idemetric, these results contrast nicely with the well known negative results of Kleinberg concerning efficient decentralised algorithms for finding short paths: for precisely the same model as Kleinberg's negative results hold, we are able to show that very efficient (and decentralised) algorithms exist if one allows for reasonable preprocessing.   
For deterministic distributed routing algorithms we are also able to obtain results proving that less routing information is required for idemetric graphs than in the worst case in order to achieve stretch less than 3 with high probability: while  $\Omega(n^2)$ routing information is required in the worst case for stretch strictly less than 3 on almost all pairs, for idemetric graphs the total routing information required is  $O(nlog(n))$.
\end{abstract}



\noindent $^{1}$State Key Laboratory of Computer Science, 
 Institute of Software, Chinese Academy of Sciences, Beijing, 100190, P. R. China. 

\noindent $^{2}$School of Computer Science, University of Chinese Academy of Sciences, Beijing, P. R. China. 

\noindent $^{3}$Department of Mathematics, London School of Economics, London, UK. 

\noindent $^{4}$State Key Laboratory of Software Development Environment, School of Computer Science, BeiHang University, 100191, Beijing,  P. R. China.

\noindent $^{5}$Department of Computer Science, University of Toronto, Canada. 

\noindent $^{6}$Department of Computer Science, Stanford, USA.


\thanks{}

\vfill \thispagestyle{empty}
\clearpage

\section{Introduction}
One of the basic tasks of network science is to identify properties which are common to most real-world networks of interest, from telecommunication and computer networks, to networks arising in biological or social contexts. A well known example of such a property is given by the fact that these networks tend to be \emph{small worlds}, i.e.\ nodes in real-world networks tend to have short paths between them. 
The phrase ``six degrees of separation'', which originated with the experiments of Stanley Milgram \cite{SM67} in the 60s, sums up the idea that two people normally have a short path of contacts connecting them. This idea has subsequently been extensively verified \cite{MK89,WS98,AJB99} and has since permeated popular culture.

    As well as the existence of short paths, however, a notable feature of many of these networks is that they have \emph{strongly concentrated distance distributions}. Most distances, in other words, are almost the same. Close to 60\% of Facebook user pairs are at distance 5 in the underlying friendship network \cite{RH16}, for example,  while over 80\% of actor pairs on the IMDB\footnote{The data for the Internet Movie Database can be downloaded at http://www.imdb.com/interfaces} are at a distance either 4 or 5.\footnote{The graph considered here is the same as that in the so called \emph{Kevin Bacon game}: nodes are actors or actresses, and there is an edge between two nodes if they have appeared in the same film.}  Perhaps to an even greater extent, this is a feature which also carries over to most small-world network models. Figure \ref{graphs1} shows examples of distance distributions for three of the best known small-world models, the Erd\H{o}s-R\'{e}nyi \cite{ER60}, Preferential Attachment \cite{BA99} and Watts-Strogatz \cite{WS98} models. As we shall see, in these and most other small-world models, such as  the  Chung-Lu \cite{CL02,CL03} and  Norros-Reittu \cite{NR06} models, the distance distributions can be \emph{proved} to become proportionately more concentrated (in a certain formal sense) as the size of the network increases.  
    
    In order to study this phenomenon in a formal setting, of course we need a mathematical definition which makes precise the idea that ``most distances are almost the same''. Definition \ref{iddef} below is new, and seems natural.

\begin{figure}[h!]
\includegraphics[scale=0.47]{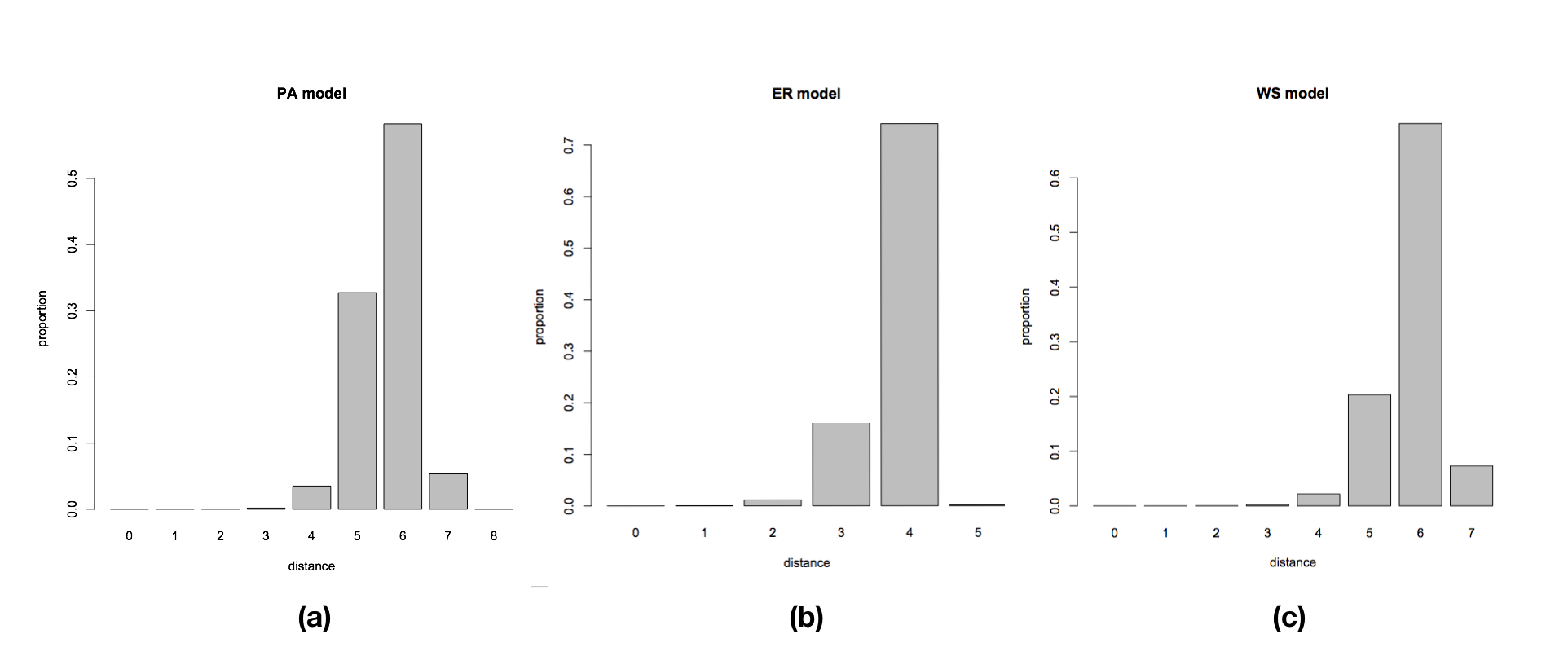}
\caption{These graphs show examples of distance distributions for three small-world network models. Distance is displayed on the $x$-axis, while the $y$-axis shows the proportion of pairs at that distance: (a) is for the Preferential Attachment (PA) model, with $10^6$ nodes and mean degree 6; (b) is for the Erd\H{o}s-R\'{e}nyi model $G(n,p)$, with $n=5\times 10^4$ nodes and expected degree 25; (c) is for the Watts-Strogatz model with $10^6$ nodes, in which each node is initially connected to 10 neighbours either side, and the rewiring probability is 0.2. }
 \label{graphs1}
\end{figure} 

 \noindent       Let $\mathbb{G}$ be a random network model for which distances are defined in the standard fashion,\footnote{ If $a$ and $b$ are nodes in a (possibly directed) graph, then we let $d(a,b)$ denote the length of a shortest path from $a$ to $b$ (where all edges are traversed in the forward direction if the graph is directed, and with $d(a,b)=\infty$ if no such path exists). In the case of a weighted graph, the above definition still applies, if the length of a path is defined to be the sum of all edge weights in the path. Note that in a directed graph it may not be the case that $d(a,b)=d(b,a)$. In an undirected graph, however, $d$ is a metric.} which produces for every $n\in \mathbb{N}$ (and perhaps taking other inputs, which we consider for now to be fixed) an ensemble of graphs with $n$ nodes, i.e.\ a probability distribution over a certain set of graphs with $n$ nodes.  For each $n$ we let $G_n$ be a graph sampled from the distribution given by $\mathbb{G}$. 
Recall that if $X$ is a random variable and $\{ X_n \}_{n\geq 0}$ is a sequence of random variables, then we say that $X_n$ tends to $X$ \emph{in probability}, denoted $X_n \xrightarrow{\enskip P \enskip} X$, if for every $\epsilon>0$, $ \mathbb{P}(|X_n-X|>\epsilon)\rightarrow 0$ as $n\rightarrow \infty$. 

\begin{defi} \label{iddef} We say  $\mathbb{G}$ is \textbf{idemetric} if there exists a finite valued function $f$ (i.e. $f(n)\neq\infty$),  such that if   $u_n$ and $v_n$ are nodes chosen uniformly at random from $G_n$, then $X_n=d(u_n,v_n)/f(n) \xrightarrow{\enskip P \enskip} 1$.\footnote{To be clear, $X_n$ is defined by first choosing $G_n$ according to the distribution given by $G$, and then choosing $u_n$ and $v_n$ uniformly at random from $G_n$.}
\end{defi} 

\noindent   For network models in which the graph produced may not be connected but does have a giant component, a weaker notion than idemetric is also useful. We shall say an event occurs \emph{with high probability} if it occurs with probability tending to 1 as $n\rightarrow \infty$.

\begin{defi} 
We say that a random network model $\mathbb{G}$ is \textbf{weakly idemetric} if both: 
\begin{enumerate} 
\item There exists $\kappa\in (0,1]$ such that, with high probability, the largest component of $G_n$  is of size at least $\kappa n$. 
\item There exists a function $f(n)$ such that if   $u_n$ and $v_n$ are nodes chosen uniformly at random from  the same largest component of $G_n$, then $X_n=d(u_n,v_n)/f(n)  \xrightarrow{\enskip P \enskip} 1$.
\end{enumerate} 
\end{defi} 

 \noindent Let us say that a network model is \emph{partially unbounded}  if there exists $\epsilon>0$ such that for all  $b$ the following holds for sufficiently large  $n$: if   $u_n$ and $v_n$ are nodes chosen uniformly at random from $G_n$, then $d(u_n,v_n)>b$ with probability $>\epsilon$.
 Many of the network models  which have been shown to be idemetric (without the term itself being used)  actually satisfy a property which is stronger, providing the model is partially unbounded:

\begin{defi} We say  $\mathbb{G}$ is \textbf{strongly idemetric}  (SI) if there exists a finite valued function $f$, and a bound $b_{\epsilon}$ for each $\epsilon>0$ (independent of $n$),   such that, for all sufficiently large $n$,  if   $u_n$ and $v_n$ are chosen uniformly at random from $G_n$, then $|f(n)-d(u_n,v_n)|<b_{\epsilon}$ with probability $>(1-\epsilon)$.  
\end{defi} 

So we actually have three closely related notions: weakly idemetric, idemetric and strongly idemetric. 
While the idemetric definition might seem sufficient to correctly capture the intended notion, the strongly idemetric definition will turn out to be at least as significant due to how commonly it is satisfied.  Immediate evidence that these definitions are of interest is then provided by the fact that  many of the best known small-world models have actually already been shown to be (weakly/strongly) idemetric. Classical results in graph theory (see \cite{RD06} or \cite{RH})  suffice to show, for example, that when the expected degree of each node is greater than 1, the Erd\H{o}s-R\'{e}nyi random graph $G(n,p)$ is weakly idemetric (although that term is not used in the literature). Bounded degree expanders are known to be strongly idemetric \cite{JSW09}. Similarly, without the term itself being applied, a range of inhomogeneous random graph models, such as the Preferential Attachment model of Barab\'{a}si and Albert \cite{BA99}, the Chung-Lu \cite{CL02, CL03} model and the Norros-Reittu model \cite{NR06}, have been shown to be idemetric for a wide range of parameters. In fact there are a number of ways of making the definition of the Preferential Attachment model precise, and a generalised model has often been considered (see, for example, \cite{DHH10}) which depends on two parameters, an integer $m \geq 1$ and a real $\delta > -m$: roughly speaking, in every step a new node is added to the network and connected to $m$ existing nodes with a probability proportional to their degree plus $\delta$. This produces a degree distribution with power law exponent $\tau= 3+\delta/m$, meaning that all values of $\tau>2$ are possible (which is the motivation for including $\delta$ in the definition of the model). Let the graph distance $H_n$ be the distance between two randomly chosen nodes in the giant component. When $m=1$ and $\delta>-1$, we have \cite{BP94,RH}: 

\[ \frac{H_n}{ \text{log} (n)}  \xrightarrow{\enskip P \enskip} \frac{2(1+\delta)}{2+\delta}. \]  

When $m\geq 1$ and $\delta>0$ (so that $\tau>3$) there exist $0 < a_1 < a_2 < \infty$ such that, as $n \rightarrow \infty$, $\mathbb{P}(a_1 \text{log} (n) \leq  H_n \leq  a_2 \text{log} (n)) = 1 - o(1)$. If $m \geq  2$ and $\delta = 0$ (so that $\tau =3 $) then \cite{BR04,RH}: 

\[  H_n \frac{\text{log log} (n)}{ \text{log} (n)}  \xrightarrow{\enskip P \enskip}  1. \] 

Finally, if $m \geq  2$ and $\delta \in (-m, 0)$ then \cite{DHH10, DMM12, RH}: 

\[ \frac{H_n}{\text{log log} (n)}  \xrightarrow{\enskip P \enskip} \frac{4}{| \text{log} (\tau - 2)|}. \]  

When $m\geq 2$ the graph produced is connected with high probability, meaning that these results suffice to give idemetricity. When $m=1$ the connectivity properties depend somewhat on the details of the model (such as whether self-edges are allowed), but the established results still suffice to establish that the corresponding models are at least weakly idemetric. For a detailed summary of these results together with self-contained proofs see \cite{RH}. 

Of course the Preferential Attachment model was introduced to try and explain the power law degree distributions observed in many real networks. An alternative and related approach which has been extensively studied, is to consider the configuration model applied to an i.i.d. sequence of degrees with a power-law degree distribution. According to this model one starts by sampling the degree sequence from a power law and then connects nodes with the sampled degrees at random (possibly resulting in a multigraph). The results obtained here are broadly similar, although constant multiplicative factors may differ \cite{HHM05,HHZ07}. 
When $\tau\in (2,3)$ graph distance centres around $2 \text{log log} (n)/| \text{log}(\tau - 2)|$, while for $\tau>3$  the graph distance grows logarithmically with $n$. 

In this paper we shall show further that the Watts-Strogatz model \cite{WS98}, and the Kleinberg model \cite{JK00} are both  idemetric for a wide range of parameters.

Given that it seems most (if not all) well-known small world models are at least weakly idemetric for a wide range of parameters,  we then  distinguish two clear aims. 

\begin{enumerate} 
\item In the spirit of the programme of research on quasi-random graphs by Fan Chung et al.\ \cite{CGW89,CG08}, which is aimed at showing that many conditions satisfied by random graphs are equivalent to each other, we would like to characterise the (weakly/strongly) idemetric network models. 
\item We would like to understand how one can make use of knowledge that the idemetric  property is satisfied. What is the significance for algorithm design? 
\end{enumerate}

In Section \ref{charsec} we address the first of these aims. Modulo reasonable sparsity assumptions, we establish the surprising fact that, for unweighted undirected networks, being  strongly idemetric is actually \emph{equivalent} to a very weak expander condition, which we call PUMP and which is described in Definition \ref{pump} below. This provides a direct way of providing short proofs that networks models are strongly idemetric, and we immediately apply it to establish that the Watts-Strogatz model is strongly idemetric for a wide range of parameters.  

\begin{defi} \label{ru} 
 For a node $u$, we write  $B_u(r)$ to denote the ball around $u$ of radius $r$, while $\overline{B_u(r)}$ denotes the complement of $B_u(r)$.  For sets of nodes $X,Y$, we write $e(X,Y)$ to denote the number of edges incident to one node in $X$ and one node in $Y$.
\end{defi}

\begin{defi} \label{pump}  $\mathbb{G}$ is a weak ball expander (PUMP) if whenever $0<\epsilon<\frac{1}{2}$ there exists $\alpha_{\epsilon}>0$ such that, for all sufficiently large $n$, if $u$ is chosen uniformly at random then with probability $>1-\epsilon$  both of the following hold: (i) there exists $r$ with $|B_u(r)|\geq \epsilon n$, (ii) for all $r$ with $ \epsilon n \leq |B_u(r)| \leq (1-\epsilon)n$,  $e(B_u(r), \overline{B_u(r)})\geq \alpha_{\epsilon}|B_u(r)| $. 
\end{defi} 

In Section \ref{pathf} we then address the second aim above. For idemetric network models we observe  that a single breadth-first search provides a solution to the all-pairs shortest paths problem, so long as one is prepared to accept paths which are of stretch close to 2 with high probability.\footnote{Recall that a path from $a$ to $b$ is of stretch $k$ if it is at most $k$ times as long as the shortest path.} The significance of this result stems from the broad class of networks to which it applies. One might hope, nevertheless, to be able to achieve smaller stretch, by further restricting the class of network models considered in a reasonable fashion: 

\begin{ques} Does there exist a condition which is satisfied by most/all well-known small world network models, and which allows for an efficient\footnote{By an efficient algorithm in this context we mean one with query times which are $O(\text{log}(n))$, and with preprocessing time $O(n)$ for sparse graphs.}  solution to the all-pairs shortest paths problem, with stretch $<2$? 
\end{ques} 

Since we are then able to show that Kleinberg's model is idemetric for a wide range of parameters, our results here contrast nicely with the well known negative results of Kleinberg \cite{JK00} concerning efficient decentralised algorithms for finding short paths: for precisely the same model as Kleinberg's negative results hold, we are able to show that very efficient (and decentralised) algorithms exist if one allows for reasonable preprocessing.  
 
For deterministic distributed routing algorithms we are also able to obtain results proving that less routing information is required for idemetric graphs than the worst case in order to achieve stretch less than 3 with high probability: while  $\Omega(n^2)$ routing information is required in the worst case for stretch strictly less than 3 on almost all pairs, for idemetric graphs the total routing information required is  $O(nlog(n))$.

\section{Characterising the idemetric network models} \label{charsec}

Throughout this section, we restrict attention to undirected, unweighted graphs. The results of this section thus apply to all small-world models mentioned in the paper other than the Kleinberg model, which is directed and is dealt with later.  If one is interested in real-world networks, then it makes sense to restrict attention further to graphs which are sparse, and in particular to network models in which individual nodes have finite expected degree in the limit. A natural way to formalise this is as follows. 

\begin{defi} \label{FED}
Let $\text{deg}(u)$ denote the degree of the node $u$. If $G_n$ is sampled from the distribution given by $\mathbb{G}$ and $u$ is chosen uniformly at random from $G_n$, then $\boldsymbol{p}_n(d)$ is the probability that $\text{deg}(u)=d$. 
 We'll say an undirected network model is of finite expected degree (FED) if there exists a probability distribution $\boldsymbol{p}$ on $\mathbb{N}$ with finite mean such that:
\begin{enumerate}  
\item[(F1)]  For each $d\in \mathbb{N}$, $  \text{lim}_{n\rightarrow \infty } \boldsymbol{p}_{n}(d) = \boldsymbol{p}(d)$; 
\item[(F2)] $\text{lim}_{n\rightarrow \infty } \sum_d d\boldsymbol{p}_{n}(d) = \sum_{d} d \boldsymbol{p}(d)$.
\end{enumerate} 
\end{defi}


So in Definition \ref{FED}, (F1) and (F2) just   say that $\boldsymbol{p}_{n}$ converges  nicely to $\boldsymbol{p}$ with finite mean as $n\rightarrow \infty$. 
FED is a natural sparsity condition which we can expect to be satisfied by the small-world network models we study (for `realistic' parameter values): all of the undirected network models mentioned in this paper satisfy the condition for a wide range of parameter values.   The following sparsity condition, however,  is the one we shall actually need in our proofs, and is implied by  FED.
 Let's  say a set of nodes $U$ is a \emph{node cover} for a set of edges $E$, if every edge from $E$ is incident with at least one node from $U$.

\begin{defi}   $\mathbb{G}$ is \textbf{uniformly sparse} (US)  if, for each $\epsilon >0$, there exists a constant $C_{\epsilon}$ such that the following holds with probability at least $1-\epsilon$ for all sufficiently large $n$: for any set of  $y\geq \epsilon n$ edges from $G_n$, every node covering is of size at least $y/C_{\epsilon}$. 
\end{defi} 

\begin{lem} \label{FIU} If $ \mathbb{G}$ is FED, then it is US. 
\end{lem} 
\begin{proof} 
Let $\boldsymbol{p}$ be as guaranteed by satisfaction of FED. Define $M_n = \sum_{d}d\boldsymbol{p}_n(d)$ and $M=\sum_{d}d\boldsymbol{p}(d)$. By a \emph{formal property} $C$ of nodes, we mean a set of pairs $(G_n,u)$. For fixed $n\in \mathbb{N}$, if ($G_n$ is sampled according to the distribution given by $\mathbb{G}$ and) $u$ is chosen uniformly at random from $G_n$, then we define: 
\[ M_n(C)= \sum_{d} d\cdot \mathbb{P}[(G_n,u)\in C \ \&\  \text{deg}(u)=d].\] 
So $M_n(C)$ can be thought of as the contribution to the mean $M_n$ given by $C$.  For any formal properties $C_1$ and $C_2$, note that:  
\begin{equation} \label{close0} 
M_n \geq M_n(C_1)+M_n(C_2)-M_n(C_1\cap C_2).
\end{equation}

Given  $\epsilon>0$, let $D$ be such  $\sum_{d\leq D} d\boldsymbol{p}(d) >M-\epsilon^2/3$. Let $C_1$ be the set of all pairs $(G_n,u)$ such that $u\in G_n$ is of degree at most $D$. Satisfaction of FED guarantees that, for all sufficiently large $n$, $\sum_{d\leq D} d\boldsymbol{p}_n(d) >M_n-\epsilon^2/2$. 
So, to phrase this another way,  for all sufficiently large $n$ we have: 
\begin{equation} \label{close1}  M_n(C_1)>M_n-\epsilon^2/2. \end{equation} 
Now choose $\mu>0$ to be small, and suppose towards a contradiction that there exist infinitely many $n$ for which the following condition  holds: 
\begin{enumerate} 
\item[$(\star_n)$:] With probability at least $\epsilon$, the $\lceil \mu n \rceil$ many nodes of highest degree have total degree summing to at least $\epsilon n$.
\end{enumerate}
 Let $C_2$ be the set of all pairs $(G_n,u)$ such that the $\lceil \mu n \rceil$ many nodes of highest degree in $G_n$  have total degree summing to at least $\epsilon n$, and such that $u$ is one of one of those $\lceil \mu n \rceil$ nodes of highest degree in $G_n$. When $(\star_n)$ holds, we have that: 
 \begin{equation} \label{close2}  M_n(C_2)\geq \epsilon^2. \end{equation} 
 On the other hand, since $D$ is a bound on the degree of $u$ whenever $(G_n,u)\in C_1\cap C_2$, we also have: 
 \begin{equation} \label{close3} 
 M_n(C_1 \cap C_2) \leq D\mu. 
 \end{equation}
 So long as $\mu$ is chosen sufficiently small,  (\ref{close0})--(\ref{close3}) then produce the required contradiction. We conclude that, for an appropriate choice of $\mu>0$ and for all sufficiently large $n$, it holds with probability at least $1-\epsilon$ that the   $\lceil \mu n \rceil$ many nodes of highest degree have total degree summing to less than $\epsilon n$. Since this holds for all $\epsilon>0$, and since (F2) holds, we can then strengthen this statement slightly. For each $\epsilon>0$ we can choose $K_{\epsilon}\in \mathbb{N}$ such that, for an appropriate choice of $\mu_{\epsilon}>0$, the following holds with probability at least $1-\epsilon$ for all sufficiently large $n$:    the   $\lceil \mu_{\epsilon} n \rceil$ many nodes of highest degree have total degree summing to less than $\epsilon n$ and  the total number of edges in $G_n$ is less than $K_{\epsilon}n$. 
 
  Now if the $\lceil \mu_{\epsilon} n \rceil$ many nodes of highest degree have total degree summing to less than $\epsilon n$, it clearly holds that no set of $\lceil \mu_{\epsilon} n \rceil$ many nodes can act as a node cover for any set of  $\geq \epsilon n$ many edges. US is therefore satisfied for $C_{\epsilon}=K_{\epsilon}/\mu_{\epsilon}$. 
 \end{proof}

\noindent As mentioned in the introduction, it turns out that a characterisation of the strongly idemetric (SI) network models can be given  in terms of \emph{expander graphs}, which have been studied intensively by mathematicians and computer scientists since the 1970s and have applications in  the design and analysis of communication networks,  in the theory of error correcting codes and in pseudorandomness. For background on expanders  see \cite{HLW06}.  

\begin{defi} \label{exp} 
We say $G$ with $|G|=n$ is an $\alpha$-edge expander if the following holds for any set of nodes $S$ with $|S|\leq n/2$: $e(S,\bar S)\geq \alpha |S|$.
\end{defi}

\noindent Definition \ref{exp}, however, is too strong for our purposes, since it requires nice behaviour with respect to \emph{all} $S$ with $|S|\leq n/2$.   Definition \ref{pump}, restated below,  weakens Definition \ref{exp} by requiring that the given condition holds only when $S$ is a large ball $B_u(r)$, and even then only \emph{most of the time}. Since Definition \ref{pump} only describes conditions on $S$ with $\epsilon n \leq |S| \leq (1-\epsilon)n$, we can drop the restriction that $|S|\leq n/2$ without making the definition any stronger.

\begin{repdefi}{pump}  $\mathbb{G}$ is a weak ball expander (PUMP) if whenever $0<\epsilon<\frac{1}{2}$ there exists $\alpha_{\epsilon}>0$ such that, for all sufficiently large $n$, if $u$ is chosen uniformly at random then with probability $>1-\epsilon$  both of the following hold: (i) there exists $r$ with $|B_u(r)|\geq \epsilon n$, (ii) for all $r$ with $ \epsilon n \leq |B_u(r)| \leq (1-\epsilon)n$,  $e(B_u(r), \overline{B_u(r)})\geq \alpha_{\epsilon}|B_u(r)| $. 
\end{repdefi} 

\begin{defi} \label{leastr} 
We define $r_u(\epsilon)$ to be the least $r$ for which $|B_u(r)|\geq \epsilon n$, or to be undefined if no such $r$ exists.
\end{defi}

\begin{lem} \label{oneway}
If $\mathbb{G}$ is a PUMP  and US, then it is SI. 
\end{lem} 
\begin{proof}
 For each node $u$, consider the following condition: 
 \begin{quote} 
 $\star(u,\epsilon, \alpha_{\epsilon})$: There exists $r$ with $|B_u(r)|\geq \epsilon n$, and for all $r$ with $ \epsilon n \leq |B_u(r)| \leq (1-\epsilon)n$,  $e(B_u(r), \overline{B_u(r)})\geq \alpha_{\epsilon}|B_u(r)| $.
 \end{quote} 
To establish that $\mathbb{G}$ is SI, we use a condition which is equivalent to PUMP: 
 \begin{enumerate} 
 \item[P2]:  For every $\epsilon>0$, there exists $\alpha_{\epsilon}>0$ such that, for all sufficiently large $n$,  the following holds with probability $>1-\epsilon$: $\star(u,\epsilon,\alpha_{\epsilon})$ holds for at least  $(1-\epsilon)n$ many nodes in $G_n$.   
 \end{enumerate}
 Suppose $\mathbb{G}$ satisfies P2 and US.  Given arbitrary $\epsilon>0$, it suffices to show the existence of $b_{\epsilon}$,  such that for all sufficiently large $n$,  if   $u,v,u',v'$ are nodes chosen uniformly at random from $G_n$, then $|d(u,v)-d(u',v')|<b_{\epsilon}$ with probability $>(1-\epsilon)$.   To this end, let $\epsilon_1$ be sufficiently small compared to $\epsilon$ (we shall come back and specify precisely what this means later). Let $\alpha_{\epsilon_1}$ be the constant guaranteed by satisfaction of P2 w.r.t.\ $\epsilon_1$, and write $\alpha$ to denote $\alpha_{\epsilon_1}$. Let $C_{\alpha \epsilon_1}$ be the corresponding constant guaranteed by satisfaction of US, and write $C$ for $C_{\alpha \epsilon_1}$. Finally, to complete the initial round of definitions, put $\beta:= \alpha/C$. \\
 
 Fix a node $u$ and suppose that $\star(u,\epsilon_1,\alpha)$ holds. Very roughly, the idea now is as follows.  Let $r$ be such that $\epsilon_1 n \leq |B_u(r)| \leq (1-\epsilon_1)n$. 
 The expander property PUMP means we are guaranteed that $e(B_u(r), \overline{B_u(r)}) \geq   \alpha|B_u(r)|$. So long as it holds for any set of  $y\geq \alpha \epsilon_1 n$ edges from $G_n$, that every node covering is of size at least $y/C$, it then follows that $ |B_u(r + 1)| \geq (1 + \beta)|B_u(r)|$. This exponential growth means that most nodes will be within some constant range of distances from $u$. Then a simple transitivity argument can be used to show that this constant range is the same for most nodes $u$.

 Now let us describe the details of that argument more precisely. Since $\star(u,\epsilon_1,\alpha)$ holds, any ball $B_u(r_0)$ of size $s$ with $\epsilon_1 n \leq s \leq   (1-\epsilon_1)n$ has least $\alpha s$ edges coming out of it (i.e.\ $e(B_u(r),\overline{B_u(r)})\geq \alpha s$), which are incident with at least $\beta s$ distinct nodes outside $B_u(r)$ -- so long, that is,  as it holds for any set of  $y\geq \alpha \epsilon_1 n$ edges from $G_n$, that every node covering is of size at least $y/C$. $B_u(r_0 + 1)$ therefore has at least $(1 + \beta)s$ nodes. Iterating this argument, ball $B_u(r_0 + r)$ has at least $\text{min}((1 + \beta)^rs, (1-\epsilon_1)n)$ nodes, for any $r\in \mathbb{N}$. In particular, letting $r_0 = r_u(\epsilon_1)$, it follows that:
\begin{equation} \label{total}
 r_u(1-\epsilon_1) - r_u(\epsilon_1) \leq   \text{log}_{1+\beta}\left( \frac{1-\epsilon_1}{\epsilon_1}\right).
\end{equation}
Define $b_{\epsilon}:= 3 (\text{log}_{1+\beta}(1-\epsilon_1) - \text{log}_{1+\beta}(\epsilon_1))$. To present the transitivity argument referred to above, let $E_1$ be the set of all unordered node pairs $\{ u,v \}$ such that $\star(u,\epsilon_1,\alpha)$ and $\star(v,\epsilon_1,\alpha)$ both hold, and such that $r_u(\epsilon_1)\leq d(u,v) \leq r_u(1-\epsilon)$. Note that, so long as $\star(u,\epsilon_1,\alpha)$ holds for at least  $(1-\epsilon_1)n$ many nodes in $G_n$, there are at most $\epsilon_1 n^2 +2\epsilon_1n^2=3\epsilon_1 n^2$ many node pairs which do not belong to $E_1$. If at least $18\epsilon_1 n$ many nodes had degree less than $\frac{2}{3}n$ in $E_1$, i.e.\ for at least $18\epsilon_1 n$ many nodes $u$ there were less than $\frac{2}{3}n$ nodes $v$ with $\{ u,v \} \in E_1$, this would imply the existence of more than  $18\epsilon_1 n \cdot \frac{1}{3}n$ ordered node pairs not belonging to $E_1$, meaning more than $3\epsilon_1n^2$ unordered node pairs which don't belong to $E_1$. We conclude that less than  $18\epsilon_1 n$ many nodes have degree less than $\frac{2}{3}n$ in $E_1$. Define $E_2$ to be all those node pairs $\{ u,v \}\in E_1$, for which both $u$ and $v$ have degree at least $\frac{2}{3}n$ in $E_1$, so that there are less than $3\epsilon_1 n^2 + 18 \epsilon_1 n^2= 21 \epsilon_1 n^2$ many node pairs which do not belong to $E_2$. The point of defining $E_2$ this way is: 

\begin{enumerate} 
\item[$(\dagger)$]
If $\{ u,v \}, \{ w,x \} \in E_2$, then both $u$ and $w$ have degree at least $\frac{2}{3}n$ in $E_1$, meaning that there exists $y$ with $\{ u,y\} , \{ w,y\}\in E_1.$
\end{enumerate} 
Note that for any two node pairs that are adjacent in $E_1$, i.e.\ for any two node pairs in $E_1$ of the form $\{ u,v\} , \{ v,w\} $, the difference between their distances is at most $\frac{1}{3}b_{\epsilon}$, by (\ref{total}).  More generally, it then follows from $(\dagger)$, that if $\{ u,v\} , \{ w,x\} \in E_2$, the difference between their distances is at most $b_{\epsilon}$. 

To finish the proof, let us consider how $\epsilon_1$ should be defined so that $b_{\epsilon}$ suffices for the satisfaction of SI w.r.t.\ $\epsilon$. For all sufficiently large $n$, we were given that with probability at most $\epsilon_1$, it fails to be the case that $\star(u,\epsilon,\alpha_{\epsilon})$  holds for at least  $(1-\epsilon_1)n$ many nodes in $G_n$.  For all sufficiently large $n$, there is also probability at most $\alpha \epsilon_1$ that it fails to be the case that, for any set of  $y\geq \alpha \epsilon_1 n$ edges from $G_n$, every node covering is of size at least $y/C$. We can assume that $\alpha <1$. When we choose the two pairs  $\{ u,v\} $ and $\{ u',v'\} $, in order to ensure that $|d(u,v)-d(u',v')|<b_{\epsilon}$, we require both of these pairs to be in $E_2$. Since there are $n(n-1)/2$ unordered pairs, this happens with probability at least $\left(1- \frac{2\cdot 21\epsilon_1 n^2}{n(n-1)} \right)^2$. So for all sufficiently large $n$ the probability that at least one of the pairs is not in $E_2$ is at most $85 \epsilon_1$. It thus suffices that $\epsilon_1 <\epsilon/87$.  
  \end{proof}

\begin{lem} \label{theother}
 If $\mathbb{G}$ is SI and US, then it is a PUMP. 
\end{lem} 
\begin{proof}
Given US, suppose that the condition PUMP fails to hold. So there exists $\epsilon_0>0$, such that for all $\alpha>0$ there exist infinitely many $n$ for which the following holds with probability at least $\epsilon_0$: for $u$ chosen uniformly at random  $\star(u,\epsilon_0,\alpha)$ fails to hold.    It suffices to show that if $\epsilon$ is chosen small enough compared to $\epsilon_0$, then, given any $b_{\epsilon}\in \mathbb{N}$ there exist infinitely many $n$  such that if   $u,v,u',v'$ are nodes chosen uniformly at random from $G_n$ then, with probability $\geq\epsilon$, either one of the distances $d(u,v), d(u'v')$ is infinite or else $|d(u,v)-d(u',v')|>b_{\epsilon}$. Choose $\epsilon <(\epsilon_0)^2/4$  and suppose given $b_{\epsilon}>0$. 
For any node $u$, let $N_u$ be the  $n\epsilon_0/2$ nodes closest to $u$ (with ties broken arbitrarily) and let $F_u$ be the $n\epsilon_0/2$ nodes furthest from $u$.  If  $\star(u,\epsilon_0,\alpha)$ fails then this happens either for reason (1), that there does not exist $r$ with $|B_u(r)|\geq \epsilon_0 n$, or for reason (2), that there exists $r$ with  $ \epsilon_0 n \leq |B_u(r)| \leq (1-\epsilon_0)n$, such that  $e(B_u(r), \overline{B_u(r)})< \alpha |B_u(r)| $. In the case that  $\star(u,\epsilon_0,\alpha)$ fails for reason (1), $v$ is then chosen such that $d(u,v)$ is infinite with probability $>1-\epsilon_0$, which we can assume is greater than $\epsilon_0/4$. 
It remains to show that if $\alpha>0$ is chosen small enough, and if $\star(u,\epsilon_0,\alpha)$ fails for reason (2), then with probability at least $1/2$, all nodes in $N_u$ are at a distance greater than $2 b_{\epsilon}$ from all nodes in $F_u$: In that case there exists $A\in \{ N,F \}$ for which  we are guaranteed $|d(u,v)-d(u',v')|>b_{\epsilon}$ so long as $v$ is chosen in $A_u$.

To establish the claim we make use of US. For a set of nodes $U$, we let $B_U(r)$ denote $\cup_{u\in U} B_u(r)$. For any $\epsilon_1>0$, let $C_{\epsilon_1/2}>1$ be the constant guaranteed by satisfaction of US with respect to $\epsilon_1/2$, and define $f(1,\epsilon_1)=\epsilon_1/2C_{\epsilon_1/2}$. So long as a node covering of any set of $n\epsilon_1/2$ many edges must be of size at least $nf(1,\epsilon_1)\ (<n\epsilon_1/2)$, we conclude that $U$ must be of size at least $nf(1,\epsilon_1)$ in order for it to be the case that $B_U(1) \geq n\epsilon_1$.  Now we can iterate this idea. For any given $\epsilon_1>0$ and $k\in \mathbb{N}$, let $\beta_0$ be such that $\epsilon_1/f(k,\epsilon_1)=\beta_0^k$. Let $C_{f(k,\epsilon_1)/2}$ be the constant guaranteed by satisfaction of US with respect to $f(k,\epsilon_1)/2$, and let $\beta_1$ be the maximum of $\beta_0$ and $2C_{f(k,\epsilon_1)/2}$. Then we define $f(k+1,\epsilon_1)= \epsilon_1/\beta_1^{k+1}$. Now define  $\alpha =f(2 b_{\epsilon},\epsilon_0/4)$, where  we can assume $b_{\epsilon}\in \mathbb{N}$.  US ensures that for all sufficiently large $n$, the following occurs with probability $>1-\alpha$: 
\begin{enumerate} 
\item[$(\dagger_{\alpha})$] For all $U$, if $|U|<\alpha n$ then $|B_U(2 b_{\epsilon})|<\epsilon_0 n/4$.
\end{enumerate}
  If  $\star(u,\epsilon_0,\alpha)$ fails for reason (2), then there exists $r$ with $\epsilon_0 n\leq |B_u(r)| \leq (1-\epsilon_0)n$ for which $|B_u(r+1)-B_u(r)|<\alpha n$, which,  so long as $(\dagger_{\alpha})$ holds, means that $r_u(1-\epsilon_0/2)>r+2 b_{\epsilon}\geq r_u(\epsilon_0/2)+2 b_{\epsilon}$. Since we can assume $\alpha<1/2$, this suffices to establish the required claim, that with probability at least $1/2$, all nodes in $N_u$ are at a distance greater than  $2 b_{\epsilon}$ from all nodes in $F_u$. 
\end{proof} 

Lemmas \ref{FIU}, \ref{oneway} and \ref{theother} give us the promised  characterisation. 

\begin{thm}  \label{characterise} 
If $\mathbb{G}$ is FED, then it is strongly idemetric iff it is a PUMP. 
\end{thm}

Along with the Erd\H{o}s-R\'{e}nyi random graph and the Preferential Attachment model of Barab\'{a}si and Albert, the Watts-Strogatz model  is one of the best known and most studied small-world models. 
The definition is given in the Appendix, and depends on two parameters $p$ and $m$: roughly, $p$ is the rewiring probability, while $m$ is the number of neighbours on each side prior to rewiring.   

Using the characterisation given by Theorem \ref{characterise}, we get the following result. The proof appears in the Appendix.   While the proof is simpler in the case that $pm>1$, this condition can presumably be significantly weakened.

\begin{thm} \label{WS} The Watts-Strogatz model is strongly idemetric when $pm>1$.  
\end{thm}

\section{Path finding}  \label{pathf} 
In this section, we consider how satisfaction of the idemetric property may be useful for algorithm design, and in particular for path finding. 

 The task of finding the shortest path between two nodes in a (possibly weighted and/or directed) graph is a fundamental problem in computer science, which has been extensively studied since the 1950s. Since often one would like to make queries for multiple pairs of nodes in the same graph, a number of variants of the problem also become significant:

\emph{The Single Source Shortest Paths (SSSP) Problem}.  The well known algorithm of Dijkstra \cite{ED59} was originally formulated to give the shortest path between a single pair of nodes, but the version which is now better known solves the SSSP problem:   a single node is fixed as the ``source'' node and the algorithm then finds shortest paths from the source to all other nodes in the graph.  For a graph with $n$ nodes and $m$ edges, this algorithm for the SSSP problem terminates in time $O(n^2)$, but has been repeatedly improved upon.  Thorup's algorithm\footnote{Here and throughout the paper we work under standard RAM model assumptions.} \cite{MT99}, for example,  runs in time $O(m)$ for undirected weighted graphs with non-negative integer weights. For undirected unweighted graphs, a simple breadth-first search suffices, terminating in time $O(n)$ for sparse graphs.

\emph{The All Pairs Shortest Paths (APSP) Problem}. For the APSP problem one is required to output a data structure encoding all shortest paths between any pair of nodes. When  a pair of nodes $(a,b)$ is queried, the data structure should return a shortest path from $a$ to $b$ in time $O(\ell)$,  where $\ell$ is the number of edges on this path. The classic Floyd-Warshall algorithm \cite{RF62,SW62} solves this problem in time $O(n^3)$ for directed graphs.  Running Thorup's solution to the SSSP for each node in an undirected graph gives an $O(nm)$ algorithm, which is much more efficient when the graph is sparse. For dense directed graphs, on  the other hand,  Williams' algorithm \cite{RW14} is the best known, running in time $O(n^3/2^{\Theta(\text{log}n)^{1/2}})$.

As pointed out by Thorup and Zwick \cite{TZ05}, however, there are many contexts in which the above solutions to the APSP problem are not satisfactory. Quite simply, the time and space complexity bounds provided by these algorithms are not practical for many real-world graphs of interest. One would like algorithms for which the preprocessing time -- the time required to produce the data structure -- is close to linear in $n$, or ideally,  which can be run efficiently in an online and decentralised fashion, responding to changes in the graph structure as they occur. It therefore becomes natural to consider ways in which one can reasonably make the problem easier, and thereby obtain more efficient solutions. One option is to accept \emph{approximate} solutions, i.e.\ algorithms which produce paths which are close to optimal.  This can be made precise by requiring paths of small \emph{stretch}, where a path from $a$ to $b$ is of stretch $k$ if it is at most $k$ times as long as the shortest path.  
Another way in which to make the problem easier is to restrict attention to graphs with convenient properties. This will be a reasonable thing to do, so long as we consider properties which one can expect to be satisfied by graphs arising from real-world networks.  Of course, our interest here is in the extent to which the restriction to idemetric networks us useful.


 \subsection{Finding short paths in idemetric networks} \label{sp1} For now let us restrict attention to unweighted graphs, although similar arguments can be made for the weighted case. If $\mathbb{G}$ is idemetric then we can obtain an approximate solution to the APSP problem very simply as follows. For each $n$ let $G_n$ be a graph generated  according to $\mathbb{G}$ with $n$ nodes, and let $r_n$ be a \emph{beacon} which is chosen uniformly at random in $G_n$. We can then carry out a breadth-first search from  $r_n$, and then again with edge directions reversed if the graph is directed,  in order to find, for all $a\in G_n$, a shortest path $p(a,r_n)$ from $a$ to $r_n$ and a shortest path $p(r_n,a)$  from $r_n$ to $a$. In network models with the small-world property, each node $a$ can then store these short paths $p(a,r_n)$ and $p(r_n,a)$, taking space at most $O(\text{polylog}(n))$ for each node. For $a_n$ and $b_n$ chosen uniformly at random from the nodes in $G_n$, let $\ell_n$ be the length of the path from $a_n$ to $b_n$ given by concatenating $p(a_n,r_n)$ and $p(r_n,b_n)$. Since $\mathbb{G}$ is idemetric, we then have that: 
\[ \frac{\ell_n}{d(a_n,b_n)}\rightarrow 2 \text{ in probability. } \]

\noindent We therefore have: 
\begin{obs} \label{ob} If $\mathbb{G}$ is idemetric  then the all-pairs shortest path problem can be reduced to the single-source shortest path problem, so long as one is prepared to accept solutions which are of stretch $2+o(1)$ with high probability.
\end{obs} 
 Note that if $\mathbb{G}$ is also of small diameter (as is the case for the Watts-Strogatz and PA models, for example), then even in the small proportion of cases where stretch $2 +o(1)$ fails, a short path will still be given by the simple algorithm above. If  $\mathbb{G}$ is weakly idemetric, then similar results hold if we restrict to nodes in the giant component. The significance of Observation \ref{ob} stems from the broad class of networks to which it applies. Greedy routing, for example, can often be very efficient in networks which come with an appropriate geometric embedding, see for example \cite{BK17}. The attraction of Observation \ref{ob}, however, is that it can be applied in scenarios where there is no given geometric embedding (and where it is not even clear how greedy routing would be defined). 

It is also worth observing, that if we restrict attention to graphs of small diameter ($O(\text{polylog}(n))$ say), then one can run a decentralised version of the algorithm above, which may be seen as a distributed version of the Bellman-Ford algorithm, and which stills runs quite efficiently in an online fashion with the nodes computing in parallel. For the sake of simplicity, let's consider graphs of bounded degree. Rather than unilaterally choosing a single beacon, it suffices to have each node choose to be a beacon with probability $\text{log}(n)/n$. Then we can consider a stage by stage process, such that during each stage $s+1$ and for each beacon $r$, each node $u$ compares the shortest path $u_s(r)$ it has previously seen from $u$ to $r$, with each of the paths given by taking the edge to a neighbour $v$  and then following $v_s(r)$ (with information as to which nodes are beacons disseminating simultaneously). If the graph is of diameter $\ell$, then correct values are obtained after at most $\ell$ many of these update stages, or $\ell$ many stages after any changes to the graph. Each stage involves each node passing $O(\text{polylog}(n))$ many bits to each of its neighbours.

 The question of finding short paths in small-world graphs was addressed by Kleinberg in \cite{JK00}.  Since he considered decentralised algorithms in the absence of any preprocessing, however, his results were largely negative -- the central point of that paper was that while short paths may exist, this does not mean that they can be easily found. For a variant of the the Watts-Strogatz model, now referred to as the Kleinberg model, he was able to prove that, while short ($O(\text{log}(n))$) paths will exist between all pairs of nodes with high probability,  no efficient decentralised algorithms exist for finding short paths (in the absence of preprocessing, and for a wide range of parameter inputs). 
 To contrast with Kleinberg's negative results, we prove next that the Kleinberg model is idemetric. So  for precisely the graphs which Kleinberg is able to conclude efficient decentralised algorithms do not exist in the absence of preprocessing, we establish that very efficient decentralised algorithms do exist when reasonable levels of (even decentralised) preprocessing are permitted. We are able to conclude this simply because the network model is idemetric.  First of all, let us define the model.


\noindent \emph{The Kleinberg model}. The model we consider is exactly the same as that in Kleinberg \cite{JK00}. For any square number $n$,  we begin with a set of nodes identified with the lattice points in a $\sqrt{n}\ \times\sqrt{n}$ square, $\{(x,y) : x \in  \{1,2,...,\sqrt{n}\},y \in  \{1,2,...,\sqrt{n}\}\}$.  
The \emph{lattice distance} between two nodes $(x,y)$ and $(r,s)$ is defined to be the number of lattice steps separating them when we fix periodic boundary conditions: $d_{\ell}((x, y), (s, t))$ is the minimum value of $ |s' - x'| + |t' - y'|$, where the minimum is taken over all values $x'=x\pm\sqrt{n}, y'=y\pm\sqrt{n}, s'=s\pm\sqrt{n}, t'=t\pm\sqrt{n}$.   For a universal constant $p \geq  1$, each node $u$  has a directed edge to every other node within lattice distance $p$.  For universal constants $q \geq  0$ and $r \geq  0$, we also construct directed edges from $u$ to $q$ other nodes (the long-range contacts) using independent random trials; the $i$th directed edge from $u$ has endpoint $v$ with probability proportional to $[d_{\ell}(u,v)]^{-r}$ (to obtain a probability distribution, we divide this quantity by the appropriate normalising constant\footnote{For the case $r=0$, we fix the convention that $0^0=1$, so that the long-range outbound contact of a node $u$ could be itself.}). This defines the network model $\mathbb{KL}(r,p,q)$. Note that the lattice distance between two nodes $d_{\ell}(a,b)$ may be quite different than $d(a,b)$.

Since the Kleinberg model gives directed graphs, we cannot apply Theorem \ref{characterise} in order to prove that it is idemetric. We shall give a more direct proof, which is also more informative since it allows us to deduce the precise function $f$ (see Definition \ref{iddef}) with respect to which the condition of being idemetric holds. As well as its relationship to Observation \ref{ob}, the fact that the Kleinberg model is idemetric, is of significant interest in its own right. 

\begin{thm} \label{basic} 
For all $p,q\in \mathbb{N}$ with $p,q\geq 1$, and for all $r\in \mathbb{R}$ with $r\in [0,2)$, the random network model $\mathbb{KL}(r,p,q)$  is idemetric. \end{thm}

\begin{proof}

We consider first the case $p=q=1$. Once we have dealt with this case, generalising to arbitrary $p,q\geq 1$ will then be straightforward. The proof for the case $p=q=1$ is split into four sections A-D, and we then consider the general case $p,q\geq 1$ in Section E.

\textbf{(A) The setup}. Let $a$ and $b$ be nodes chosen uniformly at random from amongst the nodes in $G$, a graph with $n$ nodes generated according to $\mathbb{KL}(r,p,q)$. We can consider the sets of nodes $A^{\ast}_1,A^{\ast}_2,\dots$ and $B^{\ast}_1,B^{\ast}_2 \dots$, where:

\[ A^{\ast}_i= \{ u : d(a,u)=i-1\}, \ \ \ \ \ \  \ \ B^{\ast}_i= \{ u : d(u,b)=i-1 \}. \] 
We are interested in finding the least $i$ such that $b\in A^{\ast}_i$, and so are interested in the values $A_i=|A^{\ast}_i|$ (and later will also be interested in $B_i=|B^{\ast}_i|$). Clearly $A^{\ast}_1=\{ a \}$, so $A_1=1$, and then $A^{\ast}_2$ consists of the four lattice neighbours of $a$,  together with the node which is the outbound  long-range contact of $a$. As we continue to consider $A_{i}^{\ast}$ and $A_i$ for larger values of $i$, the process is complicated, however, by potential collisions:  distinct nodes $u$ and $v$ may have outbound long-range contacts which are near to each other. The outbound long-range contact of $a$ \emph{could} be one of its four lattice neighbours, for example, or two of the lattice neighbours of $a$ could have long-range contacts which are within distance one of each other. The latter possibility could then lead to double counting in $A_4$, unless one is careful. To give a useful upper bound for $A_i$, we therefore consider a simplified process in which, roughly speaking, every long-range contact appears on a new two dimensional lattice at an infinite lattice distance from the previous node.  Formally, we can simply consider the sequence $C_i$, where $C_i=0$ for $i<0$ and, for $i\geq 0$: 
\begin{equation}  \label{Cseq} C_0=1, \  \  C_{i+1}= C_i + \sum_{j\geq 1}C_{i-j}\cdot 4j = C_i + 4C_{i-1} +8 C_{i-2} +12 C_{i-3} \cdots \end{equation}    
So, as depicted in Figure \ref{Ci}, $C_1=1, C_2=5, C_3=17, C_4=57$, and so on, and $C_i$ is the value that $A_i$ would take were it not for the possibility of collisions, as described above. 

\emph{The enumeration of} $\cup_i C_i^{\ast}$. It will also be useful to have a counterpart $C_i^{\ast}$ to $A_i^{\ast}$ and $B_i^{\ast}$, so that $C_i^{\ast}$ is a set of nodes with $|C_i^{\ast}|=C_i$. To this end, we consider 3-dimensional points $(x,y,z)$, which we shall call \emph{3-nodes}, and specify that the lattice distance between $(x,y,z)$ and $(x',y',z')$ is infinite when $z\neq z'$, and is equal to the lattice distance between $(x,y)$ and $(x',y')$ otherwise. If $a=(x_a,y_a)$, we let $C^{\ast}_1= \{ (x_a,y_a,0) \}$. In order to enumerate $C^{\ast}_{i+1}$, we take each element $u=(x,y,z)$ of $C^{\ast}_i$ in turn and proceed as follows.

\begin{enumerate} 
\item  Enumerate into $C^{\ast}_{i+1}$ all 3-node lattice neighbours of $u$ (i.e.\ those 3-nodes at lattice distance 1) which have not already been enumerated into $\cup_{j\leq i+1}C^{\ast}_j$. 
\item Let $v=(x',y')$ be the outbound long-range contact of $(x,y)$. Let $z'$ be such that no nodes with third coordinate $z'$ have been enumerated into $\cup_{j\leq i+1} C^{\ast}_j$ before, and enumerate $(x',y',z')$  into $C^{\ast}_{i+1}$ as the outbound long-range contact of $u$. 
\end{enumerate} 

We say that the 3-node $u=(x,y,z)\in C^{\ast}_i$  \emph{corresponds} to the node $(x,y)$. Note that, due to collisions, it may be the case that a node $(x,y)$ has several 3-nodes in $C^{\ast}_i$ corresponding to it. The process above specifies an enumeration of $\cup_i C_i^{\ast}$, and we also consider it to specify an enumeration of $\cup_i A^{\ast}_i$ in the obvious way: nodes in this set are enumerated in the order of their first corresponding 3-nodes in $\cup_i C_i^{\ast}$. Distances between 3-nodes are defined in the standard way, in terms of the number of edges on the shortest directed path between two nodes: there are directed edges from each 3-node $u$ to each of its lattice neighbours, i.e.\ those $v$ such that $d_{\ell}(u,v)=1$, and a directed edge from $u$ to its outbound long-range contact.  For 3-nodes $u$ and $v$, we say that $v$ is a descendant of $u$ if  $u$ is a 3 node in $C^{\ast}_i$, $v$ is a 3-node in $C^{\ast}_j$ for some $j>i$ and $d(u,v)=j-i$.  We will normally be interested in the descendants of $u$, only in the case that $u$ is a long-range outbound contact.

\begin{center} 
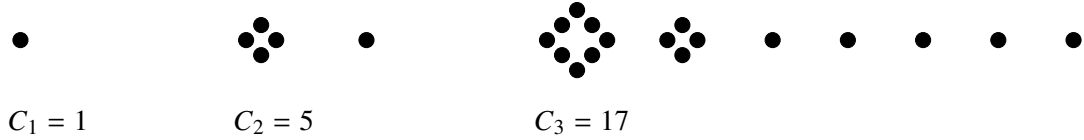
\begin{figure} 
\begin{center} 
\begin{tikzpicture}[xscale=1,yscale=1]
\draw [fill] (0,0) circle [radius=0.1];

\node[below right] at (-0.3,-0.8) {$C_1=1$}; 
\node[below right] at (2.7,-0.8) {$C_2=5$}; 
\node[below right] at (6.7,-0.8) {$C_3=17$}; 

\draw [fill] (3,0) circle [radius=0.1];
\draw [fill] (3.4,0) circle [radius=0.1];
\draw [fill] (3.2,0.2) circle [radius=0.1];
\draw [fill] (3.2,-0.2) circle [radius=0.1];

\draw [fill] (4.6,0) circle [radius=0.1];

\draw [fill] (7,0) circle [radius=0.1];
\draw [fill] (7.2,0.2) circle [radius=0.1];
\draw [fill] (7.2,-0.2) circle [radius=0.1];
\draw [fill] (7.4,0.4) circle [radius=0.1];
\draw [fill] (7.4,-0.4) circle [radius=0.1];
\draw [fill] (7.6,0.2) circle [radius=0.1];
\draw [fill] (7.6,-0.2) circle [radius=0.1];
\draw [fill] (7.8,0) circle [radius=0.1];

\draw [fill] (8.6,0) circle [radius=0.1];
\draw [fill] (9,0) circle [radius=0.1];
\draw [fill] (8.8,0.2) circle [radius=0.1];
\draw [fill] (8.8,-0.2) circle [radius=0.1];

\draw [fill] (10,0) circle [radius=0.1];

\draw [fill] (11,0) circle [radius=0.1];

\draw [fill] (12,0) circle [radius=0.1];
\draw [fill] (13,0) circle [radius=0.1];
\draw [fill] (14,0) circle [radius=0.1];

\end{tikzpicture}
\end{center} 
\caption{A schematic depiction of the growth of the sequence $\{ C_i \}_{i\in \mathbb{N}}.$ }
\label{Ci} 
\end{figure}
\end{center} 

It follows immediately from the definitions that for all $i$, $A_i\leq C_i$.  Since $C_i=C_{i-1}+\sum_{j\geq 1}C_{i-j-1}\cdot 4j$ for $i>0$, equation (\ref{Cseq}) can then be rewritten for $i>0$: 
\[  C_{i+1}= 2C_i + 3C_{i-1} + 4\sum_{j\geq 2}C_{i-j}. \] 
Given that, for $i>2$,  $C_{i-1}=2C_{i-2}+3C_{i-3}+4\sum_{j\geq 2} C_{i-j-2}$, this in turn then gives, for $i>2$: 
\[  \label{Cseq3}  C_{i+1}= 2C_i + 4C_{i-1} + 2C_{i-2}+C_{i-3}.  \] 
Standard techniques can then be applied in order to solve this linear recurrence relation. Since the characteristic polynomial 
\begin{equation} \label{char}  x^4-2x^3-4x^2-2x-1 \end{equation}  
has a largest root 
\begin{equation} \label{alpha} \alpha \approx 3.38298, \end{equation}  it follows that for some constant $\rho>0$: 
\begin{equation} \label{limC}  \text{lim}_{i\rightarrow \infty} \frac{C_i}{\rho\cdot \alpha^i} = 1. \end{equation} 
Let us define:
\[  \ell_n= \text{log}_{\alpha} n. \]
It is immediate that  $C_{i+1}\geq \sum_{j\leq i}C_j$. To within an additive constant, we then have that $\ell_n$ is the least $k$ such that $\sum_{i=1}^k C_i\geq n$. Note that, for $i>0$, $C_{i+1}\geq 2C_i$. For any $\epsilon>0$, it follows that: 
\[ 
\text{lim}_{n\rightarrow \infty} \frac{\sum_{i=1}^{(1-\epsilon)\ell_n}C_{i}}{n}=  0. 
\] 
Since $\sum_{i=1}^{(1-\epsilon)\ell_n}A_{i} \leq \sum_{i=1}^{(1-\epsilon)\ell_n}C_{i}$, we conclude that, for any $\epsilon>0$, with high probability: 
\[  d(a,b)> (1-\epsilon) \ell_n. \] 
To prove the theorem, it thus suffices to show that, for any $\epsilon>0$,  the following holds with high probability: 
\begin{equation}  \label{upbound} d(a,b)< (1+\epsilon) \ell_n. \end{equation} 

\textbf{(B) Proving (\ref{upbound})}. 
First of all we need to establish a useful bound for the distribution on long-range contacts. Suppose that we are given arbitrary nodes $u$ and $v$, but that we do not know the value of $u'$ which is the outbound long-range contact of $u$.  Let  $\text{ln}(n)$ denote the natural logarithm. Then we shall show that, irrespective of the given relative positions of $u$ and $v$, the fact that  $r\in [0,2)$ implies: 

\begin{equation} \label{lowpb} 
\mathbb{P}(u'=v)\geq   \frac{1}{4n\text{ln}(n)}.  
\end{equation} 
The proof of (\ref{lowpb}) is given in the Appendix.  Towards proving (\ref{upbound}), fix $\epsilon$ with $0<\epsilon<1/6$,  and for the remainder of the proof let $\ell^{\ast}= \frac{(1+\epsilon)\ell_n}{2}$ (note that we drop the subscript $n$ for convenience). The basic idea, as depicted in Figure \ref{lsets}, is that so long as $A^{
\ast}_{\ell^{\ast}}$ and $B^{\ast}_{\ell^{\ast}}$ are reasonably large, it is very likely the case that one of the nodes in $A^{\ast}_{\ell^{\ast}}$  has some node in $B^{\ast}_{\ell^{\ast}}$ as an outbound long-range contact. This gives a short path from $a$ to $b$. More precisely,  what we do is to show that for any constant $\mu>1$, the following both hold with high probability: 
\begin{equation} \label{lowerb} 
A_{\ell^{\ast}}>\mu  n^{1/2}\text{ln}(n)\ \ \ \ \ \ \ \text{            and             } \ \ \ \ 
 \ \ \ B_{\ell^{\ast}}>\mu n^{1/2}\text{ln}(n).
\end{equation} 
If this holds for $\mu>4$ then either $\bigcup_{i=1}^{\ell^{\ast}} A^{\ast}_i$ and $\bigcup_{i=1}^{\ell^{\ast}} B^{\ast}_i$ already have non-empty intersection, which gives the required short path from $a$ to $b$, or else we can apply (\ref{lowpb}) to conclude that the probability every member of $A^{\ast}_{\ell^{\ast}}$ will fail to have a member of $B^{\ast}_{\ell^{\ast}}$ as a long-range contact is bounded above by: 
\[ \left( 1- \frac{4 n^{1/2}\text{ln}(n)}{4n\text{ln}(n)}\right)^{\mu n^{1/2}}=  \left( 1- \frac{1}{n^{1/2}}\right)^{\mu n^{1/2}} \rightarrow e^{-\mu} \text{  as  }n\rightarrow \infty. \] 
Given that $\mu$ was arbitrary, (\ref{lowerb}) therefore suffices to give (\ref{upbound}). It remains to establish (\ref{lowerb}).

\begin{center} 
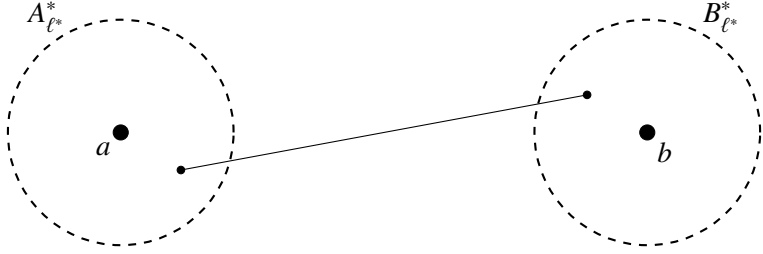
\begin{figure} 
\begin{center} 
\begin{tikzpicture}[xscale=1,yscale=1]
\draw [fill] (2,4) circle [radius=0.1];

\draw [fill] (2.8,3.5) circle [radius=0.05];

\draw [fill] (9,4) circle [radius=0.1];
\draw [fill] (8.2,4.5) circle [radius=0.05];

\draw (2.8,3.5) -- (8.2,4.5); 

\node[below left] at (2,4) {$a$}; 
\node[below right] at (9,4) {$b$}; 

\node[above left] at (1.4,5.2) {$A^{\ast}_{\ell^{\ast}}$}; 
\node[above right] at (9.6,5.2) {$B^{\ast}_{\ell^{\ast}}$};



\draw[thick,dashed] (2,4) circle (1.5cm);
\draw[thick,dashed] (9,4) circle (1.5cm);

\end{tikzpicture}
\end{center} 
\caption{If $A_{\ell^{\ast}}$ and $B_{\ell^{\ast}}$ are much larger than $n^{1/2}\text{ln}(n)$ there is likely to be an edge between them, giving a short path from $a$ to $b$. }
\label{lsets} 
\end{figure}
\end{center}

\textbf{(C) Proving (\ref{lowerb}) for $A_{\ell^{\ast}}$}. 
 We deal first with $A_{\ell^{\ast}}$ -- achieving the lower bound for $B_{\ell^{\ast}}$ then uses most of the same ideas but is complicated by the fact that nodes do not have a fixed number of inbound long-range contacts. We deal with the case for $B_{\ell^{\ast}}$ in Section D of the proof. To achieve the lower bound for  $A_{\ell^{\ast}}$, we make some new definitions. Given the enumeration of $\cup_i C_i^{\ast}$ specified previously, we say $u\in C^{\ast}_{i+1}$ is \emph{collision causing} if both:
 \begin{enumerate}
 \item[(a)] It is a long-range outbound contact of an element of $C^{\ast}_i$, and corresponds to a node in $\bigcup_{j=1}^{i+1}A^{\ast}_{j}$ which is within lattice distance $\text{log}_2(n)$ of another previously enumerated element of $\bigcup_{j=1}^{i+1}A^{\ast}_{j}$, and;
 \item[(b)] It is not the descendant of a 3-node which is already collision causing.  
 \end{enumerate} 
 Since we follow the process for less than  $\text{log}_2(n)/2$ many steps (i.e.\ $\ell^{\ast}<\text{log}_2(n)/2$), long-range outbound contacts at a lattice distance $>\text{log}_2(n)$ from all other nodes will not lead to collisions of the kind discussed previously. We then define the \emph{discounted} 3-nodes, to be all those 3-nodes in $\bigcup_{i=1}^{\ell^{\ast}} C_i^{\ast}$ which are either collision causing, or else are descendants of a collision causing 3-node. The key point of these definitions is that any two 3-nodes in  $\bigcup_{i=1}^{\ell^{\ast}} C_i^{\ast}$ which are not discounted correspond to distinct nodes in $\bigcup_{i=1}^{\ell^{\ast}} A_i^{\ast}$. 
 For each $i$, let $\mathtt{nd}(C^{\ast}_i)$ be all those 3-nodes in $C^{\ast}_i$ which are not discounted, and define $\mathtt{nd}(C_i)=|\mathtt{nd}(C^{\ast}_i)|$. 
The basic idea is now that we want to show, for all sufficiently large $n$ and for all $i\leq \ell^{\ast}$, that $\mathtt{nd}(C_i)$ is unlikely to be very much smaller than $C_i$.  Since distinct nodes in $\mathtt{nd}(C^{\ast}_i)$ correspond to distinct nodes in $A^{\ast}_i$, and since it follows from the definition of $\ell^{\ast}$ that $C_{\ell^{\ast}}=\Omega(n^{(1+\epsilon)/2})$, this will suffice to give the probabilistic lower bound for  $A_{\ell^{\ast}}$ in (\ref{lowerb}).

In order to give the probabilistic lower bound for $\mathtt{nd}(C_i)$ just discussed, we will need to bound the probability that a given long-range contact is collision causing. So suppose that we are given an arbitrary node $u$ and an arbitrary set of nodes $C$. While we know which node $u$ is, and we know the elements of $C$, suppose that we do not know the value of  $u'$, which is the outbound long-range contact of $u$.   We want to show that if  $|C|\leq n^{2/3}$, then: 
\begin{equation} \label{cpbound} 
\mathbb{P}(u'\in C) =o(1). 
\end{equation}

The proof of (\ref{cpbound}) is given in the appendix. Using  this bound, we can now provide the probabilistic lower bound for $\mathtt{nd}(C_i)$ discussed previously. Let $\alpha$ be defined as before, in (\ref{alpha}). What we want to show is that, for every $\alpha_0<\alpha$, the following holds with high probability: 
\begin{equation} \label{nond}
\mathtt{nd}(C_{\ell^{\ast}})>\alpha_0^{\ell^{\ast}}.
\end{equation} 
Given (\ref{nond}), let $\alpha_0<\alpha$ be such that $\text{log}_{\alpha}\alpha_0 > (1+\epsilon/2)/(1+\epsilon)$. Then with high probability: 
\[ A_{\ell^{\ast}}\geq \mathtt{nd}(C_{\ell^{\ast}})>
\alpha_0^{\ell^{\ast}}=(\alpha^{\text{log}_{\alpha}\alpha_0})^{(1+\epsilon)\ell_n/2}>\alpha^{(1+\epsilon/2)\ell_n/2} = n^{\frac{1}{2}+\frac{\epsilon}{4}}. \] 
For any constant $\mu>0$, the last term $n^{\frac{1}{2}+\frac{\epsilon}{4}}$ is greater than $\mu n^{1/2}\text{ln}(n)$ for all sufficiently large $n$. So this gives (\ref{lowerb}) for $A_{\ell^{\ast}}$, as required. 

To prove  (\ref{lowerb}) for $A_{\ell^{\ast}}$, it therefore remains to use (\ref{cpbound}) in order to establish (\ref{nond}).  To this end, consider the version of the recurrence relation (\ref{Cseq}) which results when, in each generation, the nodes in $C_i^{\ast}$ only have $(1-\beta)C_i$ many long-range outbound contacts (for $\beta < 1$). We'll use $E_i$ to denote the new resulting sequence of values: 
\begin{equation}  \label{modCrec} E_0=1/(1-\beta), \ \ \ \ E_{i+1}=(1-\beta)E_i + \sum_{j\geq 1} E_{i-j}\cdot 4j(1-\beta). \end{equation} 

Of course $(1-\beta)E_i$ may not be integer valued, but the sequence  given by this recurrence relation is meaningful in giving a lower bound to the number of non-discounted 3-nodes there will be in each generation, in a context where we always have some integer number  $\geq (1-\beta)\mathtt{nd}(C_i)$ of non-collision causing long-range outbound contacts for elements of $\mathtt{nd}(C^{\ast}_i)$.   The same algebraic manipulations as before can be applied, in order to reduce (\ref{modCrec}) to: 
\[ E_{i+1}=(2-\beta)E_i+(4-3\beta)E_{i-1} + (2-3\beta)E_{i-2} + (1-\beta)E_{i-3}. \] 
This is gives a characteristic equation which is a function of $\beta$: 
\begin{equation} \label{deff}  f(x,\beta)= x^4-(2-\beta)x^3 -(4-3\beta)x^2-(2-3\beta)x-(1-\beta). \end{equation}
Now in a neighbourhood of $x=\alpha$, $\beta=0$, the derivates of $f$ with respect to $x$ and $\beta$ are both positive, meaning that the largest root of $f$ is a continuous decreasing function of $\beta$ in some neighbourhood of this point.  For each $\alpha_0<\alpha$ sufficiently close to $\alpha$, it follows that there exists $\beta>0$ for which $\alpha_0$ is the largest root of $f(x,\beta)$.  Similarly, for each $\alpha_0>\alpha$ sufficiently close to $\alpha$,  there exists $\beta<0$ for which $\alpha_0$ is the largest root of $f(x,\beta)$.

So suppose given $\alpha_0<\alpha$ and $\epsilon'$ with $0<\epsilon'<0.5$. To prove (\ref{nond})  it suffices to show, for  all sufficiently large $n$, that $\mathtt{nd}(C_{\ell^{\ast}}) >\alpha_0^{\ell^{\ast}}$ with probability $>1-\epsilon'$ . Choose $\alpha_1$ and $\beta$, with $\alpha_0<\alpha_1<\alpha$, $0<\beta<0.5$, and such that $\alpha_1$ is the largest root of $f(x,\beta)$. Suppose we are given $\mathtt{nd}(C^{\ast}_i)$ for some $i<\ell^{\ast}$. For all sufficiently large $n$, the fact that  we chose $\epsilon< 1/6$ in the definition of $\ell^{\ast}$ means that: 
\begin{equation} \label{Csmall} 
\left|\cup_{j\leq \ell^{\ast}} C^{\ast}_{j}\right|<n^{2/3}. 
\end{equation} 
Then, according to (\ref{cpbound}), for sufficiently large $n$ the number of outbound long-range contacts of elements of $\mathtt{nd}(C^{\ast}_i)$ which are collision causing, is stochastically dominated by: 
\[ X \sim \text{Bin}(\mathtt{nd}(C_i),\beta/2). \]  
Applying the Chernoff bound to this stochastically dominating binomial, we conclude that for some $I_{\beta}>0$, and for all sufficiently large $n$, the probability that more than a $\beta$ proportion of the outbound long-range contacts of elements of $\mathtt{nd}(C^{\ast}_i)$ are collision causing is bounded above by: 
\[ e^{-I_{\beta}\mathtt{nd}(C_i)}. \] 

Now $\mathtt{nd}(C_i)$ is guaranteed to grow at a certain rate, since $\mathtt{nd}(C_i)\geq 4(i-1)$ for $i\geq 2$. We can therefore choose $i_0$ such that $e^{-I_{\beta}4(i_0-1)}<\epsilon'/2$, and this means that with probability 1, 
\[ e^{-I_{\beta}\mathtt{nd}(C_{i_0})}<\epsilon'/2.\] 
 We chose $\beta<0.5$. So if at most a $\beta$ proportion of the long-range contacts of elements of $\mathtt{nd}(C_{i_0}^{\ast})$ are collision causing, we have $\mathtt{nd}(C_{i_0+1})\geq 1.5 \mathtt{nd}(C_{i_0})$. Since $\epsilon'<0.5$ this gives: 
 \[ e^{-I_{\beta}\mathtt{nd}(C_{i_0+1})}<(\epsilon'/2)^{1.5}=\epsilon'\sqrt{2\epsilon'}/4<\epsilon'/4.\]  
Then so long as at most a $\beta$ proportion of the long-range contacts of elements of $\mathtt{nd}(C_{i_0+1}^{\ast})$ are collision causing, we have $ e^{-I_{\beta}\mathtt{nd}(C_{i_0+2})}<\epsilon'/8$, and so on. Define $\Pi_{\beta,\epsilon'}$ to be the following event: 
for all $i$ with $i_0\leq i <\ell^{\ast}$, at most a $\beta$ proportion of the long-range contacts of elements of $\mathtt{nd}(C_{i}^{\ast})$ are collision causing.  
The analysis above then gives, for sufficiently large $n$:
 \[ \mathbb{P}(\Pi_{\beta,\epsilon'})>1-\epsilon'.\]  
 To complete the argument, consider all those 3-nodes at a lattice distance of $i_0-1$ from $a$. These are the 3-nodes we considered above, which are guaranteed to belong to $\mathtt{nd}(C^{\ast}_{i_0})$. So long as $\Pi_{\beta,\epsilon'}$ holds, we can give a lower bound for the number of descendants of these nodes which belong to $\mathtt{nd}(C^{\ast}_{i_0+i-1})$ for each $i>0$ through the series: 
 \begin{equation} \label{Eseries}  E_{0}=4(i_0-1)/(1-\beta) \ \ \ \ E_{i+1}=(1-\beta)E_i + \sum_{j\geq 1} E_{i-j}\cdot 4j(1-\beta).\end{equation} 
 There exists a constant $\rho>0$ such that $\text{lim}_{i\rightarrow \infty}  E_i/(\rho \alpha_1^i)=1 $. Since $\alpha_0<\alpha_1$, we can choose $i_1>i_0$ such that $\mathtt{nd}(C_{\ell^{\ast}}^{\ast})>\alpha_0^{\ell^{\ast}}$ whenever $\Pi_{\beta,\epsilon'}$ holds and $\ell^{\ast}>i_1$.  This gives (\ref{nond}) as required.
 
 \textbf{(D) Proving (\ref{lowerb}) for $B_{\ell^{\ast}}$}.  As remarked on previously, establishing (\ref{lowerb}) for $B_{\ell^{\ast}}$ is complicated by the fact that nodes do not have a fixed number of inbound long-range contacts. The proof is similar to the case for $A_{\ell^{\ast}}$, however, and so Section D of the proof is given in the Appendix. 

 \textbf{(E) The case $p,q\geq 1$}.  The generalised form of the recurrence relation (\ref{Cseq}) is given as follows. For $i<0$, $C_i=0$. Then $C_0=1/q$ and, for $i\geq 0$: 
\[ C_{i+1}= qC_i + \sum_{j\geq 1} \left( 4p^2(i-1/2) +2p \right) qC_{i-j}. \]  
Using this expression for $C_i$, we get that, for $i>0$: 
\[ C_{i+1}=(q+1)C_i + (2p^2+2p-1)qC_{i-1} +\sum_{j\geq 2} 4p^2qC_{i-j}. \] 
For $i>2$, this means that $C_{i-1}=(q+1)C_{i-2} + (2p^2+2p-1)qC_{i-3} +\sum_{j\geq 2} 4p^2qC_{i-j-2}$. For $i>2$ we then have:  
\[ C_{i+1}=(q+1)C_i+ \left( (2p^2+2p-1)q +1 \right) C_{i-1} + (4p^2q-q-1)C_{i-2} +(2p^2-2p+1)qC_{i-3}. \] 
 The largest root $\alpha$ of the corresponding characteristic polynomial is the same as the largest eigenvalue of the matrix: 
 \begin{equation} \label{defM} 
\boldsymbol{M}=
  \begin{bmatrix}
    q+1 & (2p^2+2p-1)q+1 & 4p^2q-q-1 & (2p^2-2p+1)q  \\
    1 & 0 & 0 & 0 \\
     0 & 1 & 0 & 0 \\
      0 & 0 & 1 & 0     
  \end{bmatrix}
\end{equation} 
That the largest eigenvalue  $\alpha$ of $\boldsymbol{M}$ is positive and real (and occurs with multiplicity 1) follows directly from the Perron-Frobenius Theorem for non-negative matrices.  In a manner precisely analogous to the proof for $p=q=1$, we conclude that, for some constant $\rho>0$:
\[  \text{lim}_{i\rightarrow \infty} \frac{C_i}{\rho\cdot \alpha^i} = 1. \] 
The remainder of the proof then goes through with only the obvious required modifications.  The definitions of $C_i^{\ast}$ and $D_i^{\ast}$ (where $D_i^{\ast}$ is defined in Section D of the proof, in the Appendix) have to be adjusted to incorporate the increased number of neighbours, for example, and the expression $\text{log}_2(n)$ must be replaced throughout with $p\text{log}_2(n)$. The expected number of inbound long-range contacts for each node is now $q$, and distributions on the number of inbound long-range contacts must be adjusted accordingly. 
The general form of (\ref{deff}) is: 
\[ x^4-(q(1-\beta)+1)x^3 -  \left( (2p^2+2p-1)q(1-\beta) +1 \right) x^2 -(4p^2q(1-\beta)-q(1-\beta)-1)x -(2p^2-2p+1)q(1-\beta). \] 
Again we have that if $\alpha$ is the largest root then in a neighbourhood of $x=\alpha$, $\beta=0$, the derivates of $f$ with respect to $x$ and $\beta$ are both positive. 
\end{proof}

So the proof establishes that the Kleinberg model is idemetric with respect to $f=\text{log}_{\alpha} n$, where $\alpha$ is the largest eigenvalue of the matrix $\boldsymbol{M}$ given in  (\ref{defM}). 

\subsection{Memory efficient routing with small stretch} 

In this subsection we consider deterministic distributed routing algorithms, but focus instead on the total routing information required in order to achieve paths of small stretch for almost all pairs. As before we shall consider methods which give short paths even when stretch $2+o(1)$ fails, so long as the graph is of small diameter (as is the case for most small-world models which give connected graphs with high probability). The standard set-up we consider is as in \cite{GG01}, and we refer the reader to that paper for the precise details of the framework. The nodes in a graph are given labels in $\{ 1,\dots, n \}$ and each edge adjacent to a node $u$ is given an \emph{output port} number  in $\{ 1,\dots,\text{deg}(u) \}$ with respect to $u$, where $\text{deg}(u)$ is the degree of $u$.  Our results here are strong in the sense that lower bounds on routing information allow arbitrary relabelling of nodes and ports, while upper bounds assume one is given arbitrary and fixed labelings of nodes and ports.   A routing function $R$ is a pair $(P, H)$, where $P$ is the \emph{port} function and $H$ is the \emph{header} function. For any two distinct nodes $u$ and $v$, $R$ produces a path $u=u_0, \dots, u_k=v$ of nodes, along with a sequence $h_0,\dots,h_k$ of headers and a sequence $p_0,\dots, p_k$ of ports. We stipulate that $h_0=v$ and $p_0=0$.  In general, a message with header $h_i$, arriving at node $u_i$ through input port $p_i$, is forwarded to the output port $P(u_i, p_i ,h_i)= p_{i+1}$ with a new header $H(u_i, p_i, h_i)=h_{i+1}$.  Routing functions are defined as a set of local routing functions $R_u=(P_u , H_u)$, one for each node $u$.  
The memory requirement of node $u$ is the length of the smallest program that computes $R_u$. The total memory requirement for a routing function is the sum of the memory requirements of the nodes. 

The first observation to be made here, is that we can employ a similar trick to that described in Section \ref{sp1} for idemetric network models. To deal with $G_n$ we pick a beacon $r_n$ uniformly at random and perform a breadth-first search from $r_n$, defining a spanning tree $T$. At each node $u$ other than $r_n$ we simply store a port number $p(u)$ which takes them the first step on a shortest path to $r_n$. At $r_n$, meanwhile, we store $O(n \text{log} n)$ bits of information: for each other node $u$, $r_n$ stores the predecessor $v$ in $T$ together with $v$'s port number for the edge to $u$.  
For $(u,v)$ as input, $u$ begins by sending a header of the form $(v,0)$ to the node at its port $p(u)$. The last coordinate 0 indicates that we are in the phase of the routing process where a path is being followed to the beacon. The next node $v$  therefore passes on the same header, and so on, until the beacon is reached. Upon reaching $r_n$, the beacon then  finds  the sequence of port numbers $q_0,\dots, q_d$ given by $T$ and leading to $v$ if followed in turn, before passing on the header $(q_1,\dots,q_d,1)$ to the node $w$ at its output port $q_0$. Since the last coordinate 1 indicates that we are now moving away from the beacon, $w$ then passes the header $(q_2,\dots,q_d,1)$ to the next node, and so on, until  the header is the singleton $(1)$ and $v$ is reached.  This argument gives the following result, where we say that a statement holds for almost all pairs if it holds with high probability for $r_n,u,v$ chosen uniformly at random. 

\begin{thm}For idemetric networks total memory requirement $O(n\text{log}(n))$ suffices to achieve stretch $2+o(1)$  for almost all pairs. 
\end{thm} 

The next result shows that, in fact, the problem is provably easier in the idemetric case than in the general case. 

\begin{thm}[Essentially Gavoille and Gengler, \cite{GG01}] \label{LB} Total memory requirement $\Omega(n^2)$  is required for general networks, to achieve stretch $< 3$ routing for almost all pairs.
\end{thm} 
\begin{proof} This can be obtained by a modification of the proof given by Gavoille and Genger, which proves the same statement when stretch $<3$ is required for all rather than almost all pairs. In the Appendix we describe the necessary modifications.  
\end{proof} 

\section{Discussion}

Many real networks have strongly concentrated distance distributions: Figure \ref{graphs2}  gives some examples derived from the Social Computing Data Repository.  Rather than real data, however, the focus of this paper has been on \emph{network models}.  We have seen here that most (if not all) of the well known small-world network models are at least weakly idemetric for a wide range of parameters. While the terminology itself was not used, this had previously been shown for  the Erd\H{o}s-R\'{e}nyi, Preferential Attachment,  Chung-Lu, and  Norros-Reittu models, as well as for bounded degree expanders. In Sections \ref{charsec} and \ref{pathf}, we gave proofs for the Watts-Strogatz and the Kleinberg models. Modulo the sparsity condition FED, in fact, we were able to show that being strongly idemetric is equivalent to the  weak expander condition PUMP for undirected and unweighted network models. As well as providing further evidence that it is a common property, this characterisation gives an easy way of providing short proofs that network models are idemetric.  It remains, however,  to thoroughly investigate the extent to which the prevalence of (weakly/strongly) idemetric models is reflected in real-world data. It would be of considerable interest to find that it is almost universally the case (as it seems to be for models) that real-world networks expand through mechanisms which cause the distance distributions to become proportionately more concentrated as the network grows. If, on the other hand, there are significant instances in which the distance distribution is seen not to remain concentrated as the size of the network increases,  then this raises an obvious question: can one define natural small-world models, which produce realistic graphs (in the sense that they produce graphs with the standard properties associated with real-world networks, such as being scale-free, having high clustering coefficients etc) and which are \emph{not} at least weakly idemetric?

\begin{figure}[h!]
\includegraphics[scale=0.50]{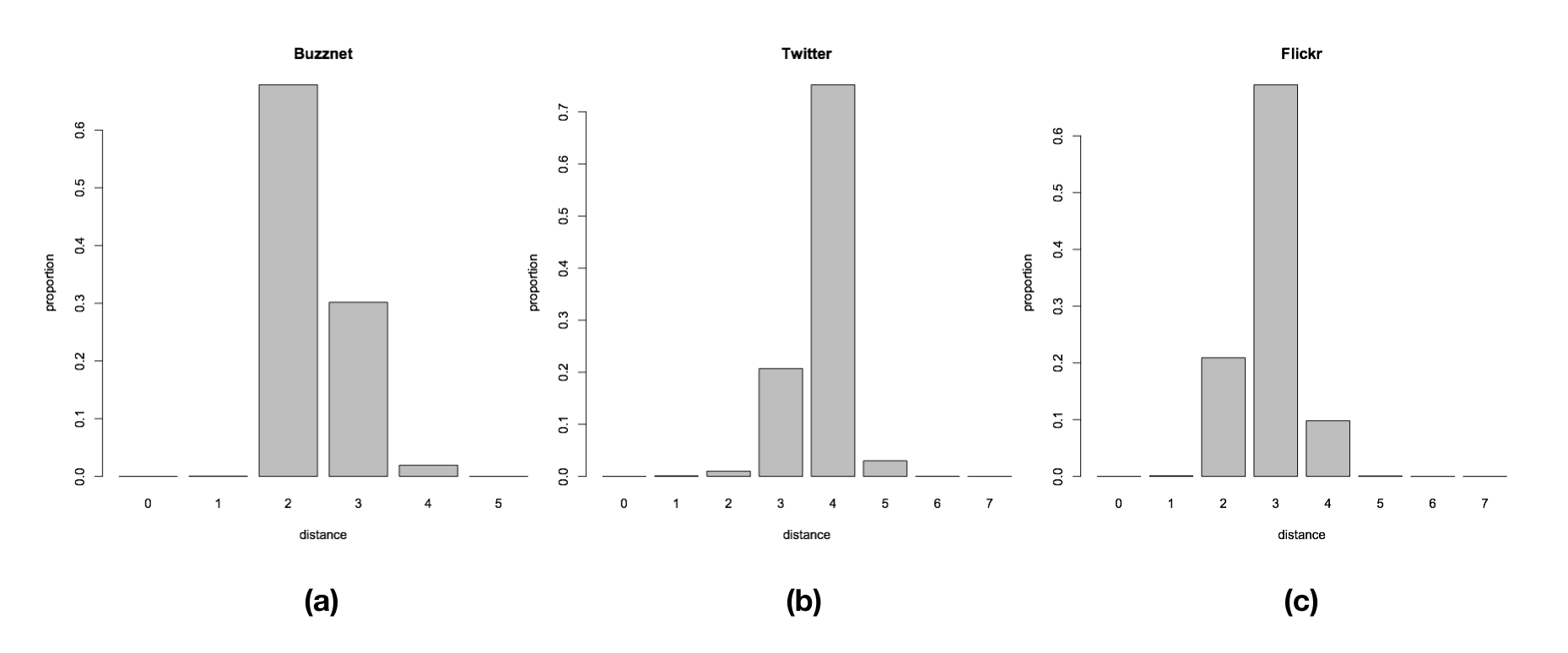}
\caption{These graphs show distance distributions for three datasets from the Social Computing Data Repository, which can be found at http://socialcomputing.asu.edu/pages/datasets. Distance is displayed on the $x$-axis, while the $y$-axis shows the proportion of pairs at that distance: (a)  Buzznet (101168 nodes, 4284534 edges); (b)  Twitter (11316811 nodes, 85331846 edges); (c) Flickr (80513 nodes, 5899882 edges). }
 \label{graphs2}
\end{figure}

\textbf{Data accessibility}. All experimental data is presented  in Figures \ref{graphs1} and \ref{graphs2} in its entirety. 

\textbf{Competing interests}. There are no competing interests to declare. 

\textbf{Authors' contributions}. All authors contributed to the construction of the proofs in this paper. 

\textbf{Acknowledgements}. None. 

\textbf{Funding statement}. Lewis-Pye was partially supported by a Royal Society University Research Fellowship. This research was partially supported by NSFC grant Nos 61772503 and 2014CB340302.  Barmpalias was supported by the 1000 Talents Program for Young Scholars from the Chinese Government No. D1101130.

\newpage

\section{Appendix} 

\subsection{The definition of the Watts-Strogatz model} \label{defWS} 

The model is defined via a two stage process. We begin with a regular ring lattice, i.e.\ $n$ nodes are arranged in a circle, such that each has an edge to the $m$ nearest nodes on each side. So if the nodes are labelled $0,\dots,n-1$, then there is an edge between distinct nodes $i$ and $j$ iff $|i-j|\leq m \text{ mod }n$.  For every node $i$, we then take each edge $(i,j)$ such that $0< (j-i) \text{ mod }n \leq m$ (i.e.\ such that $j$ is `to the right' of $i$) and `rewire' it with probability $p\in [0,1]$.  
Rewiring is done by replacing the edge $(i,j)$ with an edge $(i,k)$,  where $k$ is chosen with uniform probability from all possible values that avoid self-loops and multi-edges. Whether or not the edge is rewired, $i$ is referred to as the `fixed node' of the edge.

\subsection{The proof of Theorem \ref{WS}}

We use the same notation as in the definition of the Watts-Strogatz model in Section \ref{defWS}. So  $n$, $m$ and $p$ are as defined there. To prove the theorem, we must establish that when $pm>1$, the Watts-Strogatz model is both FED and a PUMP. 

\textbf{Establishing that the model is FED}.   When rewiring occurs and the edge $(i,j)$ is replaced with an edge $(i,k)$,  we refer to $k$  as the `new node' of the edge $(i,k)$.  If $u$ is chosen uniformly at random then $\text{deg}(u)$  can  be expressed as the sum of two terms $N(u) +F_n(u)$, where  $N(u) \sim m+\text{Bin}(m,1-p)$ and $F_n(u)$ is the number of edges for which $u$ is the new node.  

For reasons that will soon become clear, our first aim is to show that for every $\epsilon>0$, there exists a constant $c_{\epsilon}$ such that with probability $>1-e^{-c_{\epsilon}n}$:
 \begin{enumerate}
 \item[$(\dagger_{\epsilon})$] All nodes are of degree $<\epsilon n$.  
 \end{enumerate}
So suppose $\epsilon>0$. Since the total degree of all nodes is $2mn$, it immediately follows that at most a constant number $2mn/\epsilon n= 2m/\epsilon$ of nodes can have degree $\geq \epsilon n$. 
Consider the rewiring process as it proceeds around the ring, starting with node 0, then moving on to node 1, and so on. When we come to consider the rewirings for $v\neq u$, suppose that $v$ has degree $\leq \epsilon n$. In that case, no matter what has happened previously during the rewiring process, the probability that $v$ rewires an edge for which $u$ becomes the new node is less than $pm/(1-2\epsilon)n$ for all sufficiently large $n$. For all sufficiently large $n$, $F_n(u)$ is thus stochastically dominated by $X\sim 2m/\epsilon + \text{Bin}(n,pm/(1-2\epsilon)n)$. Applying Chernoff bounds, we conclude that for some constant $c$, the probability $u$ has degree $\geq \epsilon n$ is less than $e^{-cn}$. Since $u$ was chosen uniformly at random, the probability that at least one node has degree $\epsilon n$ is therefore bounded above by $n e^{-cn}$. We can thus choose $c_{\epsilon}$ as required.

 Next we want to show that $F_n(u)$  converges in distribution to $Y\sim \text{Poisson}(pm)$.\footnote{$ \text{Poisson}(pm)$ denotes a Poisson distribution with mean $pm$.}
To this end, consider again the rewiring process as it proceeds around the ring. Suppose $\epsilon>0$ and consider first the case that $\dagger_{\epsilon}$ holds. When we come to consider the rewirings for $v\neq u$, the probability that $v$ rewires an edge for which $u$ becomes the new node is less than $pm/(1-2\epsilon)n$, so long as $n$ is sufficiently large. For all sufficiently large $n$, $F_n(u)$ is thus stochastically strictly dominated by $Y_{\epsilon,n} \sim  \text{Bin}(n,pm/(1-2\epsilon)n)$ if we condition on satisfaction of $\dagger_{\epsilon}$. Since the probability that $\dagger_{\epsilon}$ fails is exponentially small, however, (and since $F_n(u)$ is always less than $mn$) we conclude that $F_n(u)$ is stochastically dominated by $Y_{\epsilon,n} \sim  \text{Bin}(n,pm/(1-2\epsilon)n)$ for all sufficiently large $n$, even  if we don't condition on satisfaction of $\dagger_{\epsilon}$. Let $\epsilon_n\rightarrow 0 $ be chosen such that $F_n(u)$ is stochastically dominated by $Y_{\epsilon_n,n}$, and define $Y_n= Y_{\epsilon_n,n}$.   As  $n\rightarrow \infty$,  $Y_n$ converges in distribution to $Y\sim \text{Poisson}(mp)$. In general, for any discrete random variables $X_n,Y_n$ and $Y$, if we have that (a) for all sufficiently large $n$, $Y_n$ stochastically dominates $X_n$, (b) every $X_n$ has mean $M$, and (c) $\mathbb{E}(Y_n)\rightarrow M$, while the sequence $Y_n$ converges in distribution to $Y$, then it follows that the sequence $X_n$ converges in distribution to $Y$. The expected degree of $u$ is $2m$, so since $\mathbb{E}(N(u))= 2m-pm$, this immediately implies  $\mathbb{E}(F_n(u))=mp$. It thus follows that $F_n(u)$ converges in distribution to $Y\sim \text{Poisson}(mp)$, as claimed. 

Let $Y\sim \text{Poisson}(mp)$, and define $\boldsymbol{p}(d)=\mathbb{P}(N(u)+Y=d)$. We have shown already that (F1)  from Definition \ref{FED} holds with respect to $\boldsymbol{p}$.  That (F2) holds is immediate, since $\sum_d d\boldsymbol{p}_{n}(d)=2m$, which is the mean of the distribution $\boldsymbol{p}$. \\

\textbf{Establishing that the model is a PUMP}. It suffices to verify that the definition is satisfied for all sufficiently small $\epsilon>0$. In order to see this, note that while the condition (i) (that there exists $r$ with $|B_u(r)|\geq \epsilon n$) is stronger for larger $\epsilon$, the fact that $\alpha_{\epsilon}>0$ means that when both (i)  and (ii) from the definition are satisfied with respect to $\epsilon>0$, (i) is also satisfied for all $\epsilon'>\epsilon$ with $\epsilon'<\frac{1}{2}$.

For each $u$ and $r\geq 1$, we let $B_r^{\ast}(u)=B_r(u)-B_{r-1}(u)$. First of all we look to establish the following lemma: 

\begin{lem}  \label{N} For all $\epsilon>0$, and $N\in \mathbb{N}$, there exists $M\in \mathbb{N}$ such that for all sufficiently large $n$,  if $u$ is chosen uniformly at random then $|B_M^{\ast}(u)|\geq N$ with probability $>1-\epsilon$. 
\end{lem} 
\begin{proof} 
Since $pm>1$, we have $m\geq 2$. For $u$ which is chosen uniformly at random from $G_n$, and $r\in \mathbb{N}$, define $C_0(u)=\{u \}$.  Given $C_r(u)$, define $C_{r+1}(u)$ to be the union of  $C_r(u)$ and the set of all nodes $j$ such that there exists $i\in C_r(u)$ and an edge $(i,j)$ for which $i$ is the fixed node. For each $r\in \mathbb{N}$, define $C^{\ast}_{r+1}(u)=C_{r+1}(u)-C_r(u)$. For $r\geq 2$, let $q_r$ be the number of nodes $i\in C_{r-2}$ which are the fixed node of at least one edge which was rewired. For each $r\geq 2$ and each possible value of $q_r$ it holds with high probability that $|C^{\ast}_r(u)|>2q_r$, and that $C^{\ast}_r(u)\subseteq B^{\ast}_r(u)$.  For fixed $N$, the probability that $q_{r}<N$ tends to 0 as $r\rightarrow \infty$ (since the rewirings are independent events). So the result holds for $M$ large compared to $N$. 
\end{proof}

Now let $\epsilon>0$ be sufficiently small that $pm(1-\epsilon)>1$.  If $|B_u(r)^{\ast}|=N$, then let $v_1,\dots,v_N$ be an enumeration of $B_u(r)^{\ast}$.  In order to put a probabilistic lower bound on $|B_u(r+1)^{\ast}|$, we can take each $v_i$ in turn,  and consider the expected value of $\rho_i$, which is the number of nodes outside $B_u(r)$ which rewire an edge for which $v_i$ is the new node, and which do not rewire any edges for which some $v_j$ with $j<i$ is the new node.  If $|B_u(r+1)|<\epsilon n$, as we come to consider $v_i$, the expected value of $\rho_i$ is greater than $pm(1-\epsilon)$, no matter the values for $\rho_j$, $j<i$. We conclude that $|B_u(r+1)^{\ast}|$ stochastically dominates $X$ which is a sum of $N$ i.i.d\ random variables, each with mean $pm(1-\epsilon)$. If we choose $\beta>1$ with $\beta <pm(1-\epsilon)$, then, applying Chernoff bounds, we get that for some $I_{\beta}>0$, either  $|B_u(r+1)|\geq \epsilon n$, or else the probability that $|B_u(r+1)^{\ast}|<\beta |B_u(r)^{\ast}|$ is bounded above by $e^{-I_{\beta}N}$. Applying the same argument, and assuming that $|B_u(r+1)^{\ast}|\geq \beta |B_u(r)^{\ast}|$, we then have that either $|B_u(r+2)|\geq \epsilon n$, or else the probability that $|B_u(r+2)^{\ast}|<\beta |B_u(r+1)^{\ast}|$ is bounded above by $e^{-I_{\beta}N\beta}$. Iterating this argument, we conclude that if  $|B_u(M)^{\ast}|=N$ then the probability that there fails to exist any $r$ for which $B_u(r)\geq \epsilon n$, is bounded above by: 
\[ 
e^{-I_{\beta}N}+e^{-I_{\beta}N\beta}+ e^{-I_{\beta}N\beta^2}+ \dots \] 
For any $\epsilon>0$, the expression above is less than $\epsilon$ for all sufficiently large $N$. So, for sufficiently small $\epsilon>0$,  (i) from the definition of PUMP holds with probability $>1-\epsilon$ for all sufficiently large $n$, by Lemma \ref{N}. In fact, a more detailed analysis of the argument above allows us to draw a stronger conclusion. If  $|B_u(M)^{\ast}|=N$, and for all $r\in [M,M']$ we have $|B_{u}(r+1)^{\ast}|/|B_{u}(r)^{\ast}| \geq \beta$, then for each $\beta_1$ with $1<\beta_1<\beta$ it holds for all sufficiently large $M'$ that $|B_u(M')^{\ast}|/|B_u(M')|>1-\beta_1^{-1}$. We conclude that for all sufficiently small $\epsilon>0$, there exists $\alpha_1>0$ such that the following holds with probability $>1-\epsilon$ for all sufficiently large $n$: there exists $r$ with $|B_u(r)| \geq \epsilon n$ and such that $|B_u(r)^{\ast}|\geq \alpha_1n$. In order to conclude the proof, suppose $|B_u(r)^{\ast}|\geq \alpha_1n$, and under the assumption that $|B_{u}(r+1)|\leq (1-\epsilon)n$ consider the number of nodes outside $B_u(r)$ which rewire an edge for which the new node is in $B_u(r)^{\ast}$. For some $J_{\epsilon}>0$ the probability that this number is not at least $\epsilon \alpha_1 n$ is bounded above by $e^{-J_{\epsilon} n}$. Condition (ii) from the definition therefore holds for all sufficiently large $n$, if we set $\alpha_{\epsilon}=\epsilon$.


 \subsection{The proof of (\ref{lowpb})}
 In order to establish (\ref{lowpb}), note that $\mathbb{P}(u'=v)$ is minimised by taking $v$ at a maximum possible lattice distance from $u$, and by taking $r$ as large as possible. So it suffices to prove the result for $r=2$, and when $v$ is at a maximum lattice distance from $u$. In that case: 

\begin{equation} \label{lrprob} \mathbb{P}(u'=v)= d_{\ell}(u,v)^{-2}/ \sum_{v'\neq u} d_{\ell}(u,v')^{-2}. \end{equation}

Then we have: 
\begin{equation} \label{lrprob2}   \sum_{v'\neq u} d_{\ell}(u,v')^{-2} \leq \sum_{j=1}^{\sqrt{n}}(4j)(j^{-2}) =4\sum_{j=1}^{\sqrt{n}}j^{-1}\leq 4(1+\text{ln}(\sqrt{n}))\leq 4\text{ln}(n), \end{equation}  
where the second inequality follows from the general bound $\sum_{j=1}^x j^{-1} \leq (1+\text{ln}(x))$. Then (\ref{lrprob}) and (\ref{lrprob2}) combine to give (\ref{lowpb}),  as required.

\subsection{The proof of (\ref{cpbound})}
Recall that we are given an arbitrary node $u$ and an arbitrary set of nodes $C$. While we know which node $u$ is, and we know the elements of $C$, we do not know the value of  $u'$, which is the outbound long-range contact of $u$.   We want to show that if  $|C|\leq n^{2/3}$, then $ \mathbb{P}(u'\in C) =o(1). $

 In order to prove (\ref{cpbound}), we define the following sets of nodes: 
\[ H_i = \{ x: \ 2^{i-1}\leq d_{\ell}(u,x)<2^i \}, \ \ \ \ \ \  I_k = \{ x:\  d_{\ell}(u,x)=k \}.  \]

The number of nodes in $\bigcup_{j=1}^i H_j$ is $2^{i+1}(2^i-1)$, so long as $n$ is sufficiently large. We say that $H_i$ is \emph{complete} for $n$ when the size of $\bigcup_{j=1}^i H_j$  takes this maximum value. Similarly, we say that $I_k$ is complete for $n$ when it is of size $4k$.  In what follows we shall be concerned with $\mathbb{P}(u'\in H_i)$ and $\mathbb{P}(u'\in I_k)$ for various $i$ and $k$.  When we consider such values, the implicit assumption will always be that $n$ is sufficiently large that $H_i$ and $I_k$ are complete for $n$. 

To establish (\ref{cpbound}), our aim is to show that, for each $r\in [0,2)$, there exists $\kappa_r>0$ such that: 

\begin{equation} \label{ratio}
\frac{\mathbb{P}(u'\in H_{i+1})}{\mathbb{P}(u'\in H_i)} \geq 2^{\kappa_r} \text{      for all sufficiently large }i. 
\end{equation} 
If $H_i$ is complete for $n$ then $|\cup_{j=1}^i H_j| =2^{i+1}(2^i-1)$. Note also that $\text{lim}_{i\rightarrow \infty} 2^{i+1}(2^i-1)/(2\cdot 4^i)=1$.  It follows that for $\ell=\text{log}_4(n)$ there exists a constant $c$ such that, for all sufficiently large $n$: 
\[ \left| \cup_{j=1}^{\frac{2}{3}\ell +c} H_j \right| >n^{2/3}, \ \text{while simultaneously } H_{\ell-c} \text{ is complete for }n. \] 
Given (\ref{ratio}), this means that (\ref{cpbound}) holds when $C$ is chosen to be a set of $n^{2/3}$ many nodes which are as near as possible to $u$, and therefore holds for all $C$ with $|C|\leq n^{2/3}$. 

To prove (\ref{ratio}), choose $\beta$ with $0<\beta<2-r$. For each $I_k\subseteq H_i$ consider $I_{2k}\cup I_{2k+1}\subseteq H_{i+1}$. For all sufficiently large $k$: 
\begin{equation} \label{rat1} \frac{\mathbb{P}(u'\in I_{2k+1})}{\mathbb{P}(u'\in I_{2k})} >2^{-\beta}. \end{equation} 
Also: 
\begin{equation}  \label{rat2} \frac{\mathbb{P}(u'\in I_{2k})}{\mathbb{P}(u'\in I_{k})} =2^{1-r}. 
\end{equation} 
Combining (\ref{rat1}) and (\ref{rat2}), we get that for all sufficiently large $k$: 
\[ \frac{\mathbb{P}(u'\in I_{2k+1}) + \mathbb{P}(u'\in I_{2k}) }{\mathbb{P}(u'\in I_{k})} >2^{2-r-\beta},  \] 
so $\kappa_r= 2-r-\beta$ suffices. 

 \subsection{(D) Proving (\ref{lowerb}) for $B_{\ell^{\ast}}$}  We begin by defining the sets $D^{\ast}_i$, which are to $B_i^{\ast}$ as $C_i^{\ast}$ is to $A_i^{\ast}$.  As we enumerate the sets $D_i^{\ast}$, it is useful to monitor whether or not we have considered the long-range inbound contacts of each node before: all nodes start as \emph{unseen} and will be labelled \emph{seen} during the enumeration once we have considered their long-range inbound contacts. 
 If $b=(x_b,y_b)$,  we let $D^{\ast}_1= \{ (x_b,y_b,0) \}$. In order to enumerate $D^{\ast}_{i+1}$, we take each element $u=(x,y,z)$ of $D^{\ast}_i$ in turn and proceed as follows.

\begin{enumerate} 
\item  Enumerate into $D^{\ast}_{i+1}$ all 3-node lattice neighbours of $u$ which have not already been enumerated into $\cup_{j\leq i+1}D^{\ast}_j$. 
\item  We divide into two cases. If $(x,y)$ is not labelled as \emph{seen} then proceed according to (a) below.
Otherwise proceed according to (b).
\begin{enumerate} 
\item[(a)] Let $v_1, \dots , v_k$ be the inbound long-range contacts of $(x,y)$, where $v_j=(x_j,y_j)$. For each $j\in [1,k]$ in turn, let $z_j$ be such that no 3-nodes with third coordinate $z_j$ have been enumerated into $\cup_{j\leq i+1} D^{\ast}_j$ before, and enumerate $(x_j,y_j,z_j)$  into $D^{\ast}_{i+1}$ as an inbound long-range contact of $u$.  Label $(x,y)$ as \emph{seen}. 
\item[(b)] Let $k$ be sampled from a distribution $\text{Bin}(n,1/n)$ (independent from all other distributions considered). Let $z_1, . . . , z_k$ be distinct and such that, for each $i \in [1, k]$, no node with third coordinate $z_j$  has been enumerated into $\cup_{j\leq i+1} D^{\ast}_j$ before. For each $i \in  [1, k]$ enumerate  $(0, 0, z_i)$  into $D^{\ast}_{i+1}$ as an inbound long-range contact of $u$.
\end{enumerate}
\end{enumerate} 
 
We let $D_i=|D_i^{\ast}|$.  The definitions we gave previously for collision causing nodes, discounted nodes and descendants are applied to the 3-nodes in $\bigcup_i D_i^{\ast}$ in the obvious way -- replacing ``outbound'' everywhere with ``inbound'', replacing $A_i^{\ast}$ with $B_i^{\ast}$, and replacing $d(u,v)$ in the definition of descendant with $d(v,u)$ (we are now interested in 3-nodes as elements of $D_i^{\ast}$ rather than $C_i^{\ast}$, so these definitions \emph{replace} rather than extend the previous ones).

In order to prove (\ref{lowerb}) for $B_{\ell^{\ast}}$, it suffices to establish the following analogue of (\ref{nond}). For every $\alpha_0<\alpha$, the following holds with high probability: 
\begin{equation} \label{nondD}
\mathtt{nd}(D_{\ell^{\ast}})>\alpha_0^{\ell^{\ast}}.
\end{equation} 
The proof would proceed much as it did for (\ref{nond}), except that there are now \emph{two} non-deterministic factors influencing  $\mathtt{nd}(D_i)$: as well as the fact that 3-nodes may or may not be discounted, they may also have from $0$ to $n$ many inbound long-range contacts ($0$ to $n-1$ if $r\neq 0$). Whereas previously it followed directly from (\ref{limC}) that (\ref{Csmall}) held for all sufficiently large $n$, now we have to provide further proof that with high probability:
\begin{equation} \label{smallD} 
\left|\cup_{j\leq \ell^{\ast}}  D^{\ast}_{j}\right|<n^{2/3}. 
\end{equation} 
Only after (\ref{smallD}) is established will we be in a position to conclude that, with high probability, the  probabilistic bound (\ref{cpbound}) can be applied at all stages $i<\ell^{\ast}$. To prove (\ref{nondD}) we therefore first have to prove that, for every $\alpha_1>\alpha$, the following holds with high probability: 
\begin{equation} \label{nondDu}
\left|\cup_{j\leq \ell^{\ast}} D^{\ast}_{j}\right|<\alpha_1^{\ell^{\ast}}.
\end{equation} 
In order to see that (\ref{nondDu}) gives (\ref{smallD}), choose $\alpha_1>\alpha$ such that $\text{log}_{\alpha} (\alpha_1) <4/(3(1+\epsilon))$. Then  (\ref{nondDu}) gives that with high probability: 
\[ \left|\cup_{j\leq \ell^{\ast}} D^{\ast}_{j}\right| < \alpha_1^{\ell^{\ast}}= (\alpha^{\text{log}_{\alpha} \alpha_1})^{\ell^{\ast}}< \alpha^{\frac{4(1+\epsilon)\ell_n}{6(1+\epsilon)}}=n^{2/3}.\]

We define $\ell^{\dagger}$ to be the minimum of $\ell^{\ast}$ and the first $i$ such that $|\cup_{j\leq i+1} D_j^{\ast}| \geq n^{2/3}$. In order to establish that with high probability $\ell^{\dagger}=\ell^{\ast}$ and (\ref{nondDu}) holds, we can argue much as in the proof of (\ref{nond}). 
So suppose given $\alpha_1>\alpha$ and $\epsilon'$ with $0<\epsilon'<0.5$. To prove (\ref{nondDu})  it suffices to show, for  all sufficiently large $n$, that $\left|\cup_{j\leq \ell^{\ast}} D^{\ast}_{j}\right| <\alpha_1^{\ell^{\ast}}$ with probability $>1-\epsilon'$ . Choose $\alpha_2$ and $\beta$, with $\alpha<\alpha_2< \alpha_1$, $0<\beta<0.5$, and such that $\alpha_2$ is the largest root of $f(x,-\beta)$, where $f$ is as defined in (\ref{deff}). Now suppose we are given $\cup_{j\leq i} D^{\ast}_j$ for some $i<\ell^{\dagger}$ --- so for some $i<\ell^{\dagger}$ we are told precisely the elements of $ D^{\ast}_j$ for each $j\leq i$ (but we are not told the value of $\ell^{\dagger}$). Before we are told the long-range inbound contacts  of a given node $u\in D^{\ast}_i$ during the enumeration of $D^{\ast}_{i+1}$,  we do not know the outbound long range contacts of any nodes other than some of those we have enumerated into $\cup_{j\leq i+1} D^{\ast}_i$ already, and  the outbound long-range contacts of all
other nodes remain uniformly and independently distributed amongst the unseen nodes. The number of inbound long-range contacts of elements of $D_i^{\ast}$  is therefore stochastically dominated by $X$, which is the sum of $D_i$ i.i.d.\ random variables, each with distribution $\text{Bin}(n, 1/(n-n^{2/3}))$.  At the same time,  the number of inbound long-range contacts of elements of  $D_i^{\ast}$ stochastically dominates $X'$, which is the sum of $D_i$ i.i.d.\ random variables, each with distribution $\text{Bin}(n-n^{2/3}, 1/n)$. Applying Chernoff bounds, 
we conclude that for some $I_{\beta}>0$, and for all sufficiently large $n$, the probability that the number of inbound long-range contacts of elements of $D^{\ast}_i$ is outside the interval $[(1-\beta) D_i, (1+\beta)D_i]$  is bounded above by\footnote{Here $I_{\beta}$ has a different but similar definition to its previous use in the proof of (\ref{nond}). Similarly $\Pi_{\beta,\epsilon}$ has a similar but different definition here in the proof of (\ref{nondDu}).}: 
\[ e^{-I_{\beta}D_i}. \] 

We can then continue to argue much as in the proof of (\ref{nond}). Let  $i_0$ be such that $e^{-I_{\beta}4(i_0-1)}<\epsilon'/4$.  We chose $\beta<0.5$. So if the number of inbound long-range contacts of elements of $D^{\ast}_i$ is at least $(1-\beta)D_i$, we have $D_{i_0+1}\geq 1.5 D_{i_0}$, which means that  $e^{-I_{\beta}D_{i_0+1}}<\epsilon'/8$, and so on. We can also choose $N_0$ which depends on $\epsilon'$ but is independent of $n$, and such that $\mathbb{P}(D_{i_0}>N_0)<\epsilon'/2$.  Define $\Pi_{\beta,\epsilon'}$ to be the following event: 
$D_{i_0}\leq N_0$ and, for all $i$ with $i_0\leq i <\ell^{\dagger}$, the number of inbound long-range contacts of elements of $D^{\ast}_i$ is in the interval $[D_i(1-\beta), D_i(1+\beta)]$. The analysis above then gives, for sufficiently large $n$:
 \[ \mathbb{P}(\Pi_{\beta,\epsilon'})>1-\epsilon'.\]  
Consider the recurrence relation: 
 \[  E_{0}=N_0/(1-\beta) \ \ \ \ E_{i+1}=(1-\beta)E_i + \sum_{j\geq 1} E_{i-j}\cdot 4j(1-\beta).\]
So long as $\Pi_{\beta,\epsilon'}$ holds, we have: 
\[ \forall i \in [i_0,\ell^{\dagger}],\ \ \ D_i \leq E_{i-i_0+1}. \] 
Since there exists $\rho>0$ for which 
\[ \text{lim}_{i\rightarrow \infty} \frac{E_{i-i_0+1}}{\rho \alpha_2^i}=1, \] 
and since $\alpha_2<\alpha_1$, we can choose $i_1$ such that: 
\begin{equation} \label{alm} \forall i\in [i_1,\ell^{\dagger}], \ \ \ \left| \cup_{j\leq i} D^{\ast}_j  \right| <\alpha_1^i. \end{equation}
We also have that  $E_{\ell^{\ast}}<n^{2/3}$ for all sufficiently large $n$, which means that so long as $\Pi_{\beta.\epsilon'}$ holds and $n$ is sufficiently large, $\ell^{\dagger}=\ell^{\ast}$.  Thus (\ref{alm}) gives (\ref{nondDu}), as required.  

With (\ref{smallD}) established, we can then prove (\ref{nondD}) in almost exactly the same way that we proved (\ref{nond}).  Suppose given $\alpha_0<\alpha$ and $\epsilon'$ with $0<\epsilon'<0.5$. To prove (\ref{nondD})  it suffices to show, for  all sufficiently large $n$, that $\mathtt{nd}(D_{\ell^{\ast}}) >\alpha_0^{\ell^{\ast}}$ with probability $>1-\epsilon'$ . Choose $\alpha_1$ and $\beta$, with $\alpha_0<\alpha_1<\alpha$, $0<\beta<0.5$, and such that $\alpha_1$ is the largest root of $f(x,\beta)$, where $f$ is as in (\ref{deff}).  Let $\ell^{\dagger}$ be the minimum of $\ell^{\ast}$ and the first $i$ such that $|\cup_{j\leq i+1} D_j^{\ast}| \geq n^{2/3}$. By (\ref{nondDu}) it holds with high probability that $\ell^{\dagger}=\ell^{\ast}$. Suppose we are given $\mathtt{nd}(D^{\ast}_i)$ for some $i<\ell^{\dagger}$. 
Then, according to (\ref{cpbound}), for sufficiently large $n$ the number of inbound long-range contacts of elements of $\mathtt{nd}(D^{\ast}_i)$ which are not collision causing, stochastically dominates $X$ which is the sum of $\mathtt{nd}(D_i)$ many i.i.d\ random variables, each with mean $1-\beta/2$.  
Applying the Chernoff bound to $X$, we conclude that for some $I_{\beta}>0$, and for all sufficiently large $n$, the probability that we fail to have at least  $(1-\beta)\mathtt{nd}(D_i)$ many long-range inbound contacts of elements of $\mathtt{nd}(D^{\ast}_i)$ which are not  collision causing is bounded above by: 
\[ e^{-I_{\beta}\mathtt{nd}(D_i)}. \] 

Now $\mathtt{nd}(D_i)$ is guaranteed to grow at a certain rate, since $\mathtt{nd}(D_i)\geq 4(i-1)$ for $i\geq 2$. We can therefore choose $i_0$ such that $e^{-I_{\beta}4(i_0-1)}<\epsilon'/2$, and this means that with probability 1, 
\[ e^{-I_{\beta}\mathtt{nd}(D_{i_0})}<\epsilon'/2.\] 
 We chose $\beta<0.5$. So if at most a $\beta$ proportion of the long-range contacts of elements of $\mathtt{nd}(D_{i_0}^{\ast})$ are collision causing, we have $\mathtt{nd}(D_{i_0+1})\geq 1.5 \mathtt{nd}(D_{i_0})$. Since $\epsilon'<0.5$ this gives: 
 \[ e^{-I_{\beta}\mathtt{nd}(D_{i_0+1})}<(\epsilon'/2)^{1.5}=\epsilon'\sqrt{2\epsilon'}/4<\epsilon'/4.\]  
Then so long as at least $(1-\beta)\mathtt{nd}(D_{i_0+1})$ many inbound long-range contacts of elements of $\mathtt{nd}(D_{i_0+1}^{\ast})$ are not collision causing, we have $ e^{-I_{\beta}\mathtt{nd}(D_{i_0+2})}<\epsilon'/8$, and so on. Define $\Pi_{\beta,\epsilon'}$ to be the following event: 
for all $i$ with $i_0\leq i <\ell^{\dagger}$, at least $(1-\beta)\mathtt{nd}(D_{i})$ many inbound long-range contacts of elements of  $\mathtt{nd}(D_{i}^{\ast})$ are not collision causing.  
The analysis above then gives, for sufficiently large $n$:
 \[ \mathbb{P}(\Pi_{\beta,\epsilon'})>1-\epsilon'.\]  
 To complete the argument, consider all those 3-nodes at a lattice distance of $i_0-1$ from $b$. These  3-nodes are guaranteed to belong to $\mathtt{nd}(D^{\ast}_{i_0})$. So long as $\Pi_{\beta,\epsilon'}$ holds, we can give a lower bound for the number of descendants of these nodes which belong to $\mathtt{nd}(D^{\ast}_{i_0+i-1})$ for each $i\leq \ell^{\dagger} -i_0+1$ through the series: 
 \begin{equation}   E_{0}=4(i_0-1)/(1-\beta) \ \ \ \ E_{i+1}=(1-\beta)E_i + \sum_{j\geq 1} E_{i-j}\cdot 4j(1-\beta).\end{equation} 
 There exists a constant $\rho>0$ such that $\text{lim}_{i\rightarrow \infty}  E_i/(\rho \alpha_1^i)=1 $. Since $\alpha_0<\alpha_1$, we can choose $i_1>i_0$ such that $\mathtt{nd}(D_{\ell^{\ast}}^{\ast})>\alpha_0^{\ell^{\ast}}$ whenever $\Pi_{\beta,\epsilon'}$ holds, $\ell^{\dagger}=\ell^{\ast}$  and $\ell^{\ast}>i_1$.  This gives (\ref{nondD}) as required.

\subsection{Modifications required to change the proof of \cite{GG01} into a proof of Theorem \ref{LB}}
 We assume full knowledge of the proof of Gavoille and Gengler in \cite{GG01}. In that paper, 
 Lemma 1 gives an inequality concerning the ratio between binomial coefficients, which we now replace with the following, which can also be proved using Stirling's approximation. 
 \begin{lem} \label{binbound} Suppose $\pi(q):\mathbb{N}\rightarrow \mathbb{N}$ satisfies the condition that $\pi(q)/q\rightarrow 0$ as $q\rightarrow \infty$. Then there exists a real $\gamma >1$ and a constant $d>0$ such that for all $q$ which are multiples of 5: 
 \[ 
\frac{ {q \choose 4q/5}}{ {4q/5 \choose \pi(q)}{q \choose \pi(q)}{2q/5 \choose q/5}}>d\cdot \gamma^{q} 
\] 
 \end{lem} 
 
Gavoille and Gengler  specify for each $n$ a set of graphs with $n$ nodes, where the nodes are partitioned into three sets $V= W \cup A \cup B$, where $|W|=q$ and $|A|=|B|=p$. Each possible graph is identified by a matrix $M$ which specifies which nodes in $A$ and $W$  are neighbours (by construction this suffices to give the same information for $B$ and $W$).  The basic form of the argument involves showing that, from a routing function with stretch $< 3$, $M$ can be recovered using only a small amount of extra information. The construction of $M$ from the routing function consists of three steps. 
In the first step a matrix $M'$ is produced, in which entries are port numbers given by the routing function as responses to certain appropriately defined queries. In the next step a matrix, let us call it $M_0$ (although it is not named in the presentation given in \cite{GG01}), is produced from $M'$ by setting each entry to 1 if its numeral value is unique amongst the values in that row, and setting it to 0 otherwise. A crucial feature of $M_0$ is that any entries which are set to 1 are `correct', in the sense that the corresponding entry is 1 in $M$.\footnote{There is a small oversight in \cite{GG01}, that when the routing function for $a_i$ is being queried, actually the port number for $b_i$ could appear as a unique entry in the corresponding row, meaning that a non-neighbour of $a_i$ in $W$ is incorrectly marked as a neighbour. This is easily corrected with information specifying the port label for $a_i$ corresponding to the edge to $b_i$.} Now that we work with a routing function which only gives stretch $<3$ for \emph{almost all} pairs, however, the difficulty is that a matrix $M_1$ will be produced in place of $M_0$, which may not satisfy this correctness condition for 1 entries. It suffices to show that only a small amount of information is required in order to convert $M_1$ into $M_0$.  We can carry out 
this conversion in two stages.

Stage 1.  In the first stage, we take each row of $M_1$ in turn, and change any 1s which are 0s in $M_0$ into 0s. The number of port labels for nodes in $A$, means that any row of  $M_1$ has less than $c=4q/5$ many 1s. Note also that, as observed in \cite{GG01},  any row of $M_0$ has at most $2q/5$ many 0s.  In  all but $o(p)$ of the rows, specifying which 1s should be flipped requires specifying at most $o(q)$ many of the (less than $4q/5$ many) 1s in that row of $M_1$ (since stretch $<3$ holds for almost almost all pairs). For $o(p)$ of the rows, however,  we may have to specify up to $2q/5$ many 1s to change to 0. 
Overall, for some function $\pi:\mathbb{N}\rightarrow \mathbb{N}$ such that  $\pi(q)/q\rightarrow 0$ as $q\rightarrow \infty$, this gives the following  upper bound for the information required to carry out the first stage: 
\[ p\text{log}{4q/5 \choose \pi(q)} + o(p) \text{log}{q \choose 2q/5}. \] 

Stage 2. In the second stage we take the modified form of $M_1$ resulting from Stage 1 (in an abuse of notation, we shall still refer to it as $M_1$), and we change any 0s which are 1s in $M_0$ into 1s. Arguing similarly to as in Stage 1, we get a (generous) upper bound on the amount of information required in order to specify the appropriate bits to be flipped: 
\[ p\text{log}{q \choose \pi(q)} + o(p) \text{log}{q \choose 4q/5}. \] 

With these modifications in place, we need to consider how they affect the bound on $K$, which is the number of bits required to store the routing function. In \cite{GG01} the relevant bound is: 
\begin{equation} \label{oldbound}
K\geq p\text{log}\left( {q \choose 4q/5} \ \Big/ \ {2q/5 \choose q/5} \right) - O(n\text{log}(n)).
\end{equation}
In place of (\ref{oldbound}) we now have: 
\begin{equation} 
K\geq p\text{log}\left( {q \choose 4q/5} \ \Big/ {4q/5 \choose \pi(q)}{q \choose \pi(q)}{2q/5 \choose q/5} \right) -   o(p) \text{log}{q \choose 2q/5}- o(p) \text{log}{q \choose 4q/5} - O(n\text{log}(n)).
\end{equation}
By Lemma \ref{binbound}, and since $o(p) \text{log}{q \choose 2q/5}$ and $ o(p) \text{log}{q \choose 4q/5}$ are both $o(n^2)$, it still holds that $K=\Omega (n^2)$, as required.


\begin{thebibliography}{CHKW01}

\bibitem[AJB99]{AJB99}
R\'{e}ka Albert, Hawoong Jeong \& Albert-L\'{a}szl\'{o} Barab\'{a}si.
\newblock The diameter of the World Wide Web. 
\newblock {\em  Nature},  401, 130, 1999.

\bibitem[BA99]{BA99}
Albert-L\'{a}szl\'{o} Barab\'{a}si \& R\'{e}ka Albert. 
\newblock Emergence of scaling in random networks. 
\newblock{\em Science}. 286 (5439): 509--512, 1999.


\bibitem[BR04]{BR04}
B\'{e}la Bollob\'{a}s \& Oliver Riordan. 
\newblock The diameter of a scale-free random graph. 
\newblock {\em Combinatorica}, 24(1), 5-34, 2004.

\bibitem[BK17]{BK17}
Karl Bringmann, Ralph Keusch, Johannes Lengler, Yannic Maus \& Anisur R.\ Molla. 
\newblock Greedy routing and the algorithmic small-world phenomenon.
\newblock {\em 2017 ACM Symposium on Principles of Distributed Computing (PODC '17)}, Washington, DC, USA, July 25-27, 371--380, 2017

\bibitem[CGH18]{CGH18} 
Francesco Caravenna, Alessandro  Garavaglia \& Remco van der Hofstad. 
\newblock  Diameter in ultra-small scale-free random graphs. 
\newblock {\em Random Structures and Algorithms}, to appear. 

\bibitem[CG08]{CG08}
Fan Chun \& Ronald Graham. 
\newblock Quasi-random graphs with given degree sequences. 
\newblock {\em   Random Structures and Algorithms}, 32 (1), 1--19, 2008.



\bibitem[CGW89]{CGW89}
Fan Chung, Ronald Graham \& Richard Wilson. 
\newblock Quasi-random graphs. 
\newblock {\em   Combinatorica}, 9 (4), 345--362, 1989. 




\bibitem[CL02]{CL02} 
Fan Chung \& Linyuan Lu. 
\newblock The average distance in a random graph with given expected degrees.
\newblock {\em Proceedings of the National Academy of Sciences USA}, 99 (25), 15879--15882, 2002.


\bibitem[CL03]{CL03}
Fan Chung \& Linyuan Lu. 
\newblock The average distance in a random graph with given expected degrees.
\newblock {\em Internet Mathematics}, 1 (1), 91--113.

\bibitem[DMM12]{DMM12} 
Steffen Dereich, Christian M\"{o}nch  \& Peter M\"{o}rters. 
\newblock Typical distances in ultrasmall random networks.
\newblock {\em Advances  in Applied Probability}, 44(2), 583--601, 2012.

\bibitem[ED59]{ED59}
Edsger Dijkstra. 
\newblock A note on two problems in connexion with graphs. 
\newblock {\em Numerische Mathematik}, 1: 269--271, 1959.

\bibitem[DHH10]{DHH10}
Sander Dommers, Remco van der Hofstad \& Gerard Hooghiemstra.
\newblock  Diameters in preferential attachment graphs. 
\newblock {\em Journal of Statistical Physics}, 139, 72--107, 2010.


\bibitem[RD06]{RD06} 
Rick Durrett. 
\newblock Random Graph Dynamics. 
\newblock  Cambridge Series in Statistical and Probabilistic Mathematics, 
Cambridge University Press, New York,  USA, 2006.  

\bibitem[ER60]{ER60}
Paul Erd\H{o}s \& Alfr\'{e}d R\'{e}nyi. 
\newblock On the evolution of random graphs. 
\newblock {\em Publication of The Mathematic
Institute of The Hungarian Academy of Sciences}, 17--61, 1960.

\bibitem[RF62]{RF62}
Robert Floyd. 
\newblock Algorithm 97. 
\newblock{\em Communications of the Association for Computing Machinery}, volume 5, issue 6, page 345, 1962.

\bibitem[GG01]{GG01}
Cyrill Gavoille \& Marc Gengler. 
\newblock Space-Efficiency for Routing Schemes of Stretch Factor Three.
\newblock {\em Journal of Parallel and Distributed Computing}.  61, 697--687, 2001.

\bibitem[GP96]{GP96} 
Cyril Gavoille \& Stephane Perennes. 
\newblock Memory requirement for routing in distributed networks. 
\newblock {\em 15th Annual ACM symposium on Principles of Distributed Computing (PODC)}, 125-133, 1996.



\bibitem[RH16]{RH16} 
Remco van der Hofstad. 
\newblock Random Graphs and Complex Networks. Volume 1. 
\newblock Cambridge University Press, 2016. 

\bibitem[RH]{RH} 
Remco van der Hofstad. 
\newblock Random Graphs and Complex Networks. Volume 2. 
\newblock Cambridge University Press, to appear. 

\bibitem[HHM05]{HHM05}
Remco van der Hofstad, Gerard Hooghiemstra \&  Piet van Mieghem. 
\newblock  Distances in random graphs with finite variance degrees. 
\newblock {\em Random Structures \& Algorithms}, 27(1), 76--123, 2005.

\bibitem[HHZ07]{HHZ07} 
Remco van der Hofstad, Gerard Hooghiemstra \& Dmitri Znamenski.
\newblock  Distances in random graphs with finite mean and infinite variance degrees.
\newblock {\em  Electronic  Journal of  Probability}, 12(25), 703--766, 2007 

\bibitem[HLW06]{HLW06}
Shlomo Hoory, Nathan Linal \& Avi Wigderson.
\newblock Expander graphs and their applications. 
\newblock {\em Bulletin of the American Mathematical Society}, 43 (4), 439--561, 2006.

\bibitem[JK00]{JK00}
Jon Kleinberg. 
\newblock The small world phenomenon: an algorithmic perspective. 
\newblock {\em Proceedings of the 32nd ACM Symposium on Theory of Computing}, 163--170, 2000.


\bibitem[JSW09]{JSW09}
Jon Kleinberg, Aleksandrs Slivkins \& Tom Wexler. 
\newblock Triangulation and embedding using small sets of beacons. 
\newblock {\em  Journal of the ACM (JACM)},  Volume 56 Issue 6, 2009.  

\bibitem[MK89]{MK89}
Manfred Kochen, editor.
\newblock {\em The small world}. 
\newblock  Norwood, N.J. Ablex Pub.,  1989.


\bibitem[SM67]{SM67}
Stanley Milgram.
\newblock   The Small World Problem.
\newblock {\em Psychology Today}. 2: 60--67, 1967.

\bibitem[NR06]{NR06} 
Ilkka Norros \& Hannu Reittu. 
\newblock On a conditionally Poissonian graph process. 
\newblock {\em Advances in Applied Probability}, 38 (1), 59--75, 2006.

\bibitem[BP94]{BP94}
Boris Pittel. 
\newblock Note on the heights of random recursive trees and random $m$-ary search trees.
\newblock {\em Random Structures and Algorithms}, 5, 337?347, 1994.


\bibitem[MT99]{MT99}
Mikkel Thorup.
\newblock Undirected single-source shortest paths with positive integer weights in linear time. 
\newblock  {\em Journal of the Association for Computing Machinery}, 46 (3): 362--394, 1999. 

\bibitem[TZ05]{TZ05}
Mikkel Thorup \& Uri Zwick. 
\newblock Approximate distance oracles. 
\newblock {\em Journal of the Association for Computing Machinery}, 52 (1), 1--24, 2005.

\bibitem[SW62]{SW62} 
Stephen Warshall. 
\newblock A theorem on Boolean matrices. 
\newblock {\em Journal of the Association for Computing Machinery}, 9:11--12, 1962.

\bibitem[WS98]{WS98}
Duncan Watts \& Steven Strogatz. 
\newblock  Collective dynamics of `small-world' networks. 
\newblock {\em Nature},  393, 440--442, 1998.

\bibitem[RW14]{RW14}
Ryan Williams.  
\newblock Faster all-pairs shortest paths via circuit complexity.
\newblock {\em Proceedings of the 46th Annual ACM Symposium on Theory of Computing}, 664--673, 2014.


\end{thebibliography}
  \end{document}